\pdfoutput=1
\documentclass[fleqn,10pt]{wlscirep}
\usepackage[square,sort,comma,numbers]{natbib}
\usepackage[utf8]{inputenc}
\usepackage[T1]{fontenc}
\usepackage{tabularx}
\DeclareGraphicsExtensions{.eps}
\usepackage[caption=false,font=footnotesize]{subfig}
\usepackage{float}
\usepackage{graphicx}

\usepackage{tikz}  
\usetikzlibrary{arrows.meta}

\usepackage{multirow}

\usepackage{bm}

\usepackage{ulem}

\usepackage{lineno}
\usepackage{amsmath,amssymb,amsfonts}
\usepackage{bm}
\usepackage{bbm}
\usepackage{hyperref}
\usepackage{amsthm}
\usepackage{physics}
\theoremstyle{definition}
\newtheorem{theorem}{Theorem}

\newcounter{propositioncnt}
\newtheorem{proposition}[propositioncnt]{Proposition}
\newcounter{lemmacnt}
\newtheorem{lemma}[lemmacnt]{Lemma}
\newcounter{corollarycnt}
\newtheorem{corollary}[corollarycnt]{Corollary}
\newtheorem{remark}{Remark}
\newcounter{definitioncnt}
\newtheorem{definition}[definitioncnt]{Definition}

\usepackage[outdir=./]{epstopdf}

\title{Economic Capacity Withholding Bounds of Competitive Energy Storage Bidders}

\author[1]{Xin Qin}
\author[1]{Ioannis Lestas}
\author[2]{Bolun Xu}

\affil[1]{University of Cambridge, Department of Engineering, Cambridge, CB2 1PZ, UK}
\affil[2]{Columbia University, Earth and Environmental Engineering, New York, NY 10027, USA}

\begin{abstract}
    \textit{Problem definition}: Economic withholding in electricity markets refers to generators bidding higher than their true marginal fuel cost, and is a typical approach to exercising market power. However, existing market designs require storage to design bids strategically based on their own future price predictions, motivating storage to conduct economic withholding without assuming market power. As energy storage takes up more significant roles in wholesale electricity markets, understanding its motivations for economic withholding and the consequent effects on social welfare becomes increasingly vital. 
    \textit{Methodology/results}: This paper derives a theoretical framework to study the economic capacity withholding behavior of storage participating in competitive electricity markets and validate our results in simulations based on the ISO New England system. We demonstrate that storage bids can reach unbounded high levels under conditions where future price predictions show bounded expectations but unbounded deviations. Conversely, in scenarios with peak price limitations, we show the upper bounds of storage bids are grounded in bounded price expectations. Most importantly, we show that storage capacity withholding can potentially lower the overall system cost when price models account for system uncertainties.
    \textit{Managerial implications:} Our paper reveals energy storage is not a market manipulator but an honest player contributing to the social welfare. It helps electricity market researchers and operators better understand the economic withholding behavior of storage and reform market policies to maximize storage contributing to a cost-efficient decolonization. 
\end{abstract}
\begin{document}

\flushbottom
\maketitle
\section{Introduction}
\noindent 
Energy storage participants are becoming key players in wholesale electricity markets.  In the California Independent System Operator (CAISO), the capacity of utility-scale battery energy storage has already exceeded 5~GW and is projected to surpass 10~GW in the upcoming five years~\citep{caiso_bat}. Following California's lead, Texas has emerged as the second-largest market in storage capacity terms~\citep{ercot_bat}. Consequently, the role of storage in the market has naturally transitioned from ancillary service markets to price arbitrage in wholesale markets, reflecting the significant capacity influx~\citep{us2021form}. It's now crucial to comprehend the factors influencing energy storage market participation. Such understanding is essential for formulating regulatory policies, addressing market power concerns, and crafting effective price signals~\citep{sioshansi2021energy}.

Energy storage participants employ bidding strategies that fundamentally differ from those of conventional generators. While existing competitive electricity market designs incentivize conventional generators to base their bids on fuel costs~\citep{kirschen2018fundamentals}, energy storage participants must integrate future opportunity values into their bid designs~\citep{harvey2001market}. This is primarily because their immediate operations are influenced by anticipated price outcomes, especially in real-time markets that only clear for the imminent hour~\citep{zheng2023energy, chen2021pricing}. Consequently, storage entities must operate with strategic foresight~\citep{bansal2022market} in the present market framework {and adopt economic capacity withholding strategies: This involves the storage strategically choosing not to discharge at a certain time period, by submitting bid prices higher than the market clearing price of this time period, despite the price is higher enough to recover storage’s discharge cost. Instead, the storage expects to discharge at a future higher price to maximize its profit.
} 
This perspective is also reflected in CAISO's recent market power mitigation manual, which acknowledges that storage bids should account for opportunity costs and should \textit{``include the highest (predicted) price, corresponding to the storage duration of the resource"}~\citep{caiso_es2}. {This approach is further evidenced by storage's practical bidding behavior in CAISO's electricity market\citep{caiso2024daily}.}

Due to the significant volatility inherent in real-time power system operations, it is impractical for storage participants to expect designing bids using accurate price predictions~\citep{jkedrzejewski2022electricity}. A more pragmatic approach to increase profit potential is thus to incorporate uncertainty price models into bid design~\citep{zheng2022arbitraging}. Hence, a pivotal consideration emerges with the growth of storage market share: How would uncertainty models employed in storage's bidding strategy affect their profit and the operating cost of the entire system? A good understanding of this problem can provide insights into future market mechanisms and regulatory framework designs.

This paper studies the economic capacity withholding behavior of storage participating in competitive electricity markets. 
{
Our study offers the following key contributions:
\begin{enumerate}
    \item  We introduce a theoretical framework to analyze the economic capacity withholding of energy storage motivated by price uncertainties. This is the first paper to systematically study how the uncertainty model impacts storage market actions. Despite much previous literature emphasizing that storage must consider price uncertainty models in price arbitrage, there is no study from the system operator or regulator’s perspective on how these uncertainty-based bidding strategies could motivate economic withholding. 
    \item We state that for storage, uncertainty price models would motivate increased economic withholding.
    We prove that under various classes of distribution models, higher price uncertainty results in higher withholding. We also derive the storage economic withholding bounds in cases when the electricity market prices are bounded and discuss variations of this model to price distribution functions, such as the dependence on price spike occurrences.
    \item We show there exists an alignment between economic withholding and social welfare convergence. While Harvey and Hogan~\citep{harvey2001market} pointed out that storage economic withholding cannot be equivalent to market power exercise, ours is the first study examining how storage economic withholding may positively reduce the system operating cost, facilitating convergence in social welfare and improving market efficiencies. 
    \item We validate our result using an agent-based market simulation with Monte-Carlo scenarios based on the ISO New England system.
\end{enumerate}
}

{
Our theoretical framework for analyzing storage withholding behavior follows these steps: We first formally quantify the economic capacity withholding practice of battery energy storage in Subsection 3.4. 
Then in Subsection 4.1, we prove the concavity of storage value function, demonstrating that storage market power is limited, and show the convexity of storage bids. Third, we investigate the impact of uncertainty on storage withholding, including unbounded withholding in Subsection 4.2 and bounded withholding in Subsection 4.3. We further explore the relationship between storage withholding and social welfare in Subsection 4.4. Finally, we use simulations to demonstrate the theoretical findings and the impact of storage withholding on system operation in Section 5.
}

{The remainder of the paper is organized as follows.}
Section~2 reviews related literature. Section~3 formulates the market clearing and storage bidding models {and clarifies the definition of economic capacity withholding.} Section~4 presents our theoretical results. Section 5 shows the simulation results of proposed theorems, and extends the proposed theorem to practical market settings, and Section 6 concludes this paper. Detailed proofs are provided in the Appendix.

\section{Backgrounds}
\label{literature}
\subsection{Uncertainty Management in Power Systems}
Employment of uncertainty models in power systems operations is increasingly critical due to rising forecast errors contributed by progressive deployments of intermittent renewables~\citep{jafari2020power}. The operational constraints, such as generator start-up/shut-down requirements, ramp speeds, and storage limits, limit the flexibility of most generators to adapt swiftly to immediate set-point changes~\citep{kirschen2018fundamentals}. Consequently, {both system operators and market participants in} operations must anticipate and accommodate these limitations with respect to uncertainty factors in power systems.

Over the past decade, researchers have proposed various approaches to address increasing uncertainties in power system scheduling utilizing methods such as stochastic optimization~\citep{wang2015real, zhao2015data}, robust optimization~\citep{bertsimas2012adaptive,lubin2015robust}.
{Recently, risk-aware methods have garnered increased attention and research efforts~\citep{ndrio2021pricing, chattopadhyay2002unit}, largely due to the significant breakthroughs on the analytical reformulation of chance constraints~\citep{dvorkin2019chance, xu2023chance, yang2021chance}}.
While methodologically diverse, a common assumption in these frameworks is the central role of the power system operator in uncertainty management. This responsibility entails the development of stochastic models for renewable energy supply and demand, establishing risk preferences for dispatch, and integrating additional objectives and constraints into the market-clearing model.

However, solely relying on a centralized approach may not be practical. Firstly, the complexity of many uncertainty models can substantially increase the computational burden of power system optimizations,  hindering their implementation over real-world systems~\citep{roald2023power}. Secondly, in deregulated power systems, where supply and demand are balanced through market clearing, the operator must maintain a neutral stance. Centralizing uncertainty management will inevitably cause market prices to be affected by the uncertainty preference set by the power system operators and may not capitalize on the advantages of distributed decision-making and risk assessment. It is thus increasingly important to analyze uncertainties in electricity markets as strategic interactions between market participants hedge against risks~\citep{bose2019some}.

\subsection{Economic Capacity Withholding of Storage}
{
From a broader perspective, the relationship between uncertainty and withholding can be traced back to the 1950s, when storage optimization problems, such as the warehouse problem, emerged as a key area in operations research with storage participants modeled as price-takers. For instance, the deterministic warehouse problem modeled as linear programming was studied by ~\cite{charnes1955warehousing}, \cite{bellman1956warehousing}, and~\cite{dreyfus1957warehouse} in which the solution to the optimization problem was straightforward: Participants would choose to either fully stock the warehouse, completely deplete it, or leave it unchanged. These research works laid the foundation for defining non-withholding behavior. Later, paper~\cite{charnes1966decision} extended the warehouse problem to incorporate stochastic price variations, with subsequent studies exploring various extensions (e.g.,~\cite{kjaer2008valuation} and~\cite{ref_r1_wu}), demonstrating that withholding can be a beneficial behavior under uncertainties.}

{In electricity markets,} capacity withholding is a strategy where generators conserve a portion of their capacity inaccessible to market clearings{~\citep{ferc2002_withhold, ferc2020economicwithholding}}. This can manifest in two ways: physical withholding, where a generator offers less capacity than available; and economic withholding, where a generator sets its offer prices above its actual marginal cost, leading to reduced capacity clearance at a specific market clearing price~\citep{ferc2002_withhold}. A direct consequence of capacity withholding is steeper supply curves, resulting in higher market prices and, thus, higher profit of generation units despite less cleared quantity. Capacity withholding is a common practice to exercise market power~\citep{kirschen2018fundamentals}.

While identifying physical capacity withholding is straightforward based on the rated unit capacity, economic capacity withholding is significantly more sophisticated regarding energy storage.   Harvey and Hogan~\citep{harvey2001market} have acknowledged that energy-limited resources (which would include today's energy storage) may withhold capacity in certain hours due to expectations of higher future prices, i.e., opportunity cost, and such behaviors should be distinguished from the exercise of market power in that ``efficient pricing would fully utilize the energy of the unit in the highest price hours over the period of the limitation''. Yet, they remarked that the lack of perfect foresight, hence uncertainty or imperfect judgment, ``could have a larger impact on reducing supply than the exercise of market power''. They also made an important statement that economic withholding of energy-limited resources is ``necessary to the efficient and reliable operation of the electric grid''. {Research papers in natural gas storage, such as~\citep{ref_r1_sec} and~\citep{ref_r1_wu} investigated the benefits of economic capacity withholding from participant side (referred to as optimal capacity management and rolling intrinsic heuristic policy improvement), showing its benefits under price uncertainties and storage constraints. }

It is worth noting that Harvey and Hogan's 2001 study primarily referenced liquefied natural gas (LNG) and pumped hydro resources. These resources held a relatively marginal presence in the electricity markets of that era, and their market impact was limited. Moreover, their operational characteristics differ considerably from today's battery energy storage. For instance, LNG storage incurs consistent fuel costs, while pumped hydro operations must account for environmental and water supply constraints~\citep{rehman2015pumped}. Hence, given the associated modeling complexity and the limited impact in practice, the paper concluded that perhaps the most realistic option is to ``exempt them (energy limited resources) from the no economic or physical withholding standard''. 

With the notable rise of battery energy storage in recent years, many studies have delved into the market behavior of these storage resources. A considerable body of research underscores how integrating uncertainty price models into storage bid designs can improve market profits, especially in the face of market and system fluctuations~\citep{zheng2022arbitraging, savkin2016constrained, krishnamurthy2017energy, taylor2017price, zuluaga2023data, jiang2015optimal, du2016managing} {and the negative price~\citep{ref_r1_zhou}}. Storage can also seek to exercise market power like conventional generators to drive up prices and gain more profits~\citep{ye2017strategic, nasrolahpour2016bidding}. Apparently, storage can combine uncertainties and market powers to maximize their market return using more sophisticated bidding models~\citep{wang2017optimal}. {However, due to the nature of the supply curve, storage's considerations of uncertainty not only affect their own actions but will also impact wholesale market clearing outcomes. The impact of these uncertainty-motivated bidding strategies on economic withholding remains a key concern for market regulators but has not been systematically analyzed in existing studies}. Recent studies have also investigated storage bidding behaviors and their consequent market ramifications from the system's perspective~\citep{qin2023role, gao2022multiscale}. {However, these studies assume that storage relies on deterministic forecasts of future prices in bidding strategies, overlooking the role of uncertainties and the potential for economic withholding arising from such considerations.}

\section{Formulation and Preliminaries}
We first present the market clearing model considered in this study, followed by the storage bidding model considering uncertain future price predictions. We also present some results that are prerequisites to our main results, including storage bid design and {definition of economic withholding}.

\subsection{Preliminary and Assumptions} \label{subsec.prelim}
{We consider the model of a typical two-stage wholesale electricity market consisting of day-ahead markets (DAM) and real-time markets (RTM), as shown in Figure~\ref{fig:intro_market}. The DAM is a \textit{forward} and is mixed-integer linear programming that determines generators' start-up or shut-down status over the next 24 hours. Once these binary decisions are made, the system operator reruns the linear part of the problem and uses the dual variable associated with the locational power balance constraint as the day-ahead market clearing prices. The RTM is a \textit{spot} market, during which the system operator fixes the generator start-up and shut-down status based on the DAM results and only solves a single‐period linear optimization problem to economically dispatch resources as in Subsection~\ref{subsec.marketclear}. Similar to DAM, the dual variables to the power balance constraint serve as the real-time market clearing prices. 
For clearance, we provide the detailed DAM and RTM models in the Appendix~\ref{appendix.market}.}

\begin{figure}[h]%
    \centering
    \includegraphics[trim = 0mm 0mm 0mm 0mm, clip, width = .55\columnwidth]{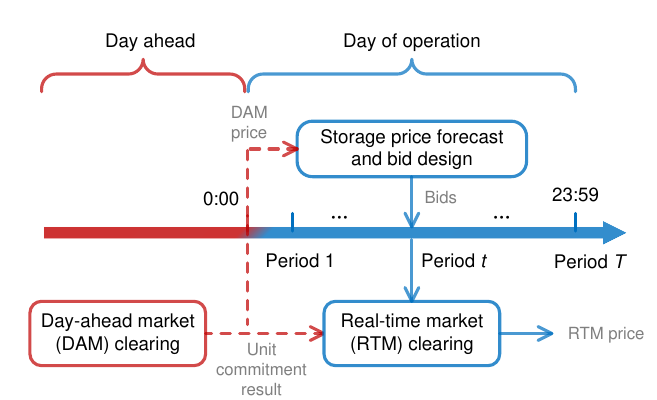}
    \caption{{Illustration of wholesale electricity market operation with storage participation in the RTM, exemplified at period $t$. Storage bid design is detailed in Subsection~\ref{subsec.esbid}, the RTM clearing problem is described in Subsection~\ref{subsec.marketclear}, and the DAM market clearing process is outlined in Appendix~\ref{appendix.market}.}}
    \label{fig:intro_market}
    \vspace{-1em}
\end{figure}

{A notable difference between DAM and RTM is that the DAM clears once over the next 24 hours, while the RTM clears at each operating time step. Then due to the need to reflect future opportunities, storage must design discharge and charge bids in RTM before the market clearing factoring future price predictions.}

{We use $\lambda_t$ to denote the RTM price. We assume storage adopts a price-taker bidding strategy. Hence, the storage participant does not seed to exercise market power and treat future prices as a stochastic process. The storage forecasts future prices} as a time-varying stage-wise independent process $\hat \lambda_t$ characterized by its expectation $\mu_t$ and {deviation} $\sigma_t$. The forecasting approach is outlined in the forth assumption below. Note that if not specified, $\hat \lambda_t$ can represent any uncertainty distribution.

To focus on the relationship between storage bidding models and system uncertainties, we consider the following assumptions.

\begin{enumerate}
    \item We consider {a single-bus system, which has many identical battery storage participants. We assume a competitive market, and participants act as price takers without intending to exercise market power individually, but they collectively influence the market.}
    \item {We consider each battery storage has a concave end-value function $V_T(e_T)$ representing the end value of energy stored. This assumption aligns with an intuitive economic sense: The more energy remains unsold at the end of a day, the less marginal value it should be.
    }
    \item {We consider storage participants adopts a finite forecast horizon (such as the end of the day), given the limited time horizon market information available to storage operators.}
    \item We consider the storage {forecasts the real-time prices $\hat \lambda_t$ based on the day-ahead market clearing price for the same time period \citep{tang2016model}, setting the day-ahead price as expectation $\mu_t$}. {This approach of real-time price forecasting is widely adopted in wholesale} markets facilitated by two factors. First, most producers and customers in electricity markets are settled in DAM, while RTM primarily settles the deviations to day-ahead settlements. Hence, real-time prices should be distributed around the day-ahead price. Second, virtual bids, in which market speculators can arbitrage between DAM and RTM, create incentives for participants to converge persistent price gaps between day-ahead to real-time.
\end{enumerate}

\subsection{Market Clearing Model} \label{subsec.marketclear}
We consider a typical two-stage energy market consisting of DAM and RTM. The DAM performs unit commitment over the next 24 hours, followed by sequential real-time dispatches.
Our study focuses on energy storage participation in RTM as shown in Figure~\ref{fig:intro_market}.
As a participant, storage submits bids to the system operator in a manner akin to a conventional generator~\citep{sakti2018review}. After receiving bids from storage and generator participants, the system operator solves the optimal dispatch problem~\eqref{eq:market} to clear the RTM, issuing both dispatch commands and releasing clearing prices.

For a time step $t\in\mathcal{T}$, the RTM clearing is thus formulated as 
\begin{subequations}\label{eq:market}
\begin{align}
    \min_{g_{t},~p_{t},~b_{t}} \; G_t(g_{t}) + {\sum_{\mathcal{S}}  O_t (p_t) - \sum_{\mathcal{S}} D_t(b_t)},
    \label{eq:obj}
\end{align}
where $\mathcal{T}= \{1,2,3,...,T\}$ indicates the set of periods, in which 1 and $T$ indicate the first and the last period of an operating day, respectively. {$\mathcal{S} = \{1,2,3,...,N_e\}$ indicates the set of storage participants, where $N_e$ denotes the number of storage participants.} $g_t$ is the aggregated power generation of conventional generators, and $G_t$ is the aggregated generation cost function. {$p_t$ and $b_t$ are the {power} discharged/charged into individual storage over time period $t$. $O_t$ and $D_t$ are the supply (discharge) offer cost function and the demand (charge) bidding function of each storage participant, respectively. Here in the theoretical analysis, cost functions $O_t(p_t) = \int o_t \ dp_t$ and $D_t(b_t) = \int d_t \ db_t$ are the integral of storage charge and charge price bids $o_t$ and $b_t$, respectively; In practice, system operators might require price-quantity segment offer curves~\citep{caiso2021energystorage}, i.e., piece-wise linear discharge and charge bids as detailed in Appendix~\ref{appendix.market}. Since the storage participants are assumed to be identical, we can omit indexing individual participants.}

The RTM includes the power balance constraint, as shown in equation~\eqref{eq:market_lam}. 
\begin{gather}
    {g_t + \sum_{\mathcal{S}} p_t - \sum_{\mathcal{S}} b_t = L_t - w_t,} \label{eq:market_lam}
\end{gather}
where $w_t$ is the renewable generation and $L_t$ is the load demand. Note that $L_t-w_t$ is the net demand and is also the source of \textit{uncertainty} in real-time power system operations. Furthermore, we define the dual variable linked to constraint~\eqref{eq:market_lam} as $\lambda_t$, which is the real-time market clearing price (real-time price).

The RTM models the storage unit constraints. Equation~\eqref{eq:obj-c1} {models the discharge power limit over period $t$, and  equation~\eqref{eq:obj-c2} models the charging power limit}
\begin{align}
    0\leq &p_t \leq \min\{P, {e_{t-1}\eta /\tau\}}, \label{eq:obj-c1}\\
    0\leq &b_t \leq \min\{P, { (E-e_{t-1}) / (\eta \tau )} \}.  \label{eq:obj-c2}
\end{align}
in which {$\tau$ is the time interval between two consecutive time periods.
$P$ is storage charging and discharging power limit, }$\eta$ indicates the storage efficiency. $e_t$ indicates storage State-of-Charge (SoC) at period $t$, and $E$ indicates full storage capacity.  Equations~\eqref{eq:obj-c1} and~\eqref{eq:obj-c2} say storage discharge energy over period $t$ should be the lesser of the limit $P$ and the energy left in the storage $e_{t-1}$; Similarly, equation~\eqref{eq:obj-c2} says storage charging energy should be the lesser of the limit $P$ and the charging headroom given full storage capacity $E$.

The market clearing is also subject to generator unit constraints, which are omitted in the main text for simplicity. For the full market clearing model formulation, including DAM and RTM that we used in our simulation, please refer to Appendix~\ref{appendix.market}.
\end{subequations}

\subsection{Storage Bidding Model} \label{subsec.esbid}
In this subsection, we elaborate on the bidding model of {an individual} storage participant, who acts as a price taker but takes price uncertainty into bidding design.

{Based on the forth assumption in Subsection~\ref{subsec.prelim},} all RTM price models explored in this paper have \textit{fixed} expectations in all bounded and unbounded cases, i.e., {RTM price expectation} $\mu_t$ is fixed. Then the storage arbitrage problem, {aiming at maximizing arbitrage profits over the considered time period,} can be formulated as stochastic dynamic programming (SDP) as follows
\begin{subequations}~\label{eq:esbid_opti}
\begin{align}
    Q_{t-1}(e_{t-1}|\hat \lambda_t) &:=  \max_{p_t, b_t}\hat \lambda_{t} (p_{t} - b_{t}) - cp_{t} + V_{t}(e_{t})\label{p1_obj}  \\
    V_{t-1}(e_{t-1}) &:= \mathbb{E}_{\hat \lambda_t }\Big[Q_{t-1}(e_{t-1}|\hat \lambda_t)\Big] \label{eq:VQ}
\end{align}
subject to
\begin{align}
    0 &\leq p_{t} \leq P \mathbf{1}_{[\hat \lambda_t \leq 0]}\label{eq:c1}\\
    0 &\leq b_{t} \leq P \label{eq:c2} \\
    e_{t} - e_{t-1} &= (-p_{t}/\eta + b_{t}\eta) \tau \label{eq:c3}\\
    0 &\leq e_{t} \leq E \label{eq:c4} 
\end{align}
\end{subequations}
where $c$ is the marginal physical discharge cost, {which is a parameter reflecting the comprehensive investment and operation cost per 1MW power discharged}. $V_{t}$ is the value-to-go function recursively defined as the optimized profit expectation in the future. {Equation~\eqref{eq:VQ} indicates $ V_{t-1}$ is the conditional expectation of $Q_{t-1}$.} \eqref{eq:c1} and \eqref{eq:c2} are the discharging and charging {power} constraints, respectively. Note in \eqref{eq:c1} that the indicator function on the upper {power} limit enforces the discharge {power} to be zero if the price is smaller than zero, hence preventing the battery from discharging at negative prices. This condition is sufficient to prevent simultaneous charge and discharge and will be an important property to utilize in our later results. \eqref{eq:c3} is the SoC efficiency constraint, and \eqref{eq:c4} is the energy capacity constraint. Also note that \eqref{eq:c1}--\eqref{eq:c4} are equivalent to \eqref{eq:obj-c1} and \eqref{eq:obj-c2}, which enforces power {subject} to efficiencies.

{
When participating in the RTM, as illustrated in Figure~\ref{fig:intro_market}, a storage participant submits discharging price bids ($o_t$) and charging price bids ($b_t$), analogous to those of a generator and a load~\citep{sakti2018review}. As a limited energy resource,  storage will incorporate both opportunity costs and physical costs in its bid design.}
{The physical cost depends on parameters like efficiency and marginal discharging costs $c$, while the opportunity cost is determined by the value function $V_t$ in~\eqref{eq:esbid_opti}. Thus, the storage discharging and charging cost functions are expressed as $C_{o,t} = c p_t - V_t(e_t)$ and $C_{d,t} = - V_t(e_t)$, respectively.} 

\begin{remark} \textbf{Storage opportunity bids.} Different from generators that bid based on physical costs, the storage design bids based on opportunity costs. For example,  as shown in equation~\eqref{eq:VQ}, storage value $V_{t-1}$ is dependent on its price prediction of the next period $\hat \lambda_{t}$ as well as the value function $V_{t}$ at period $t$. Similarly, $V_{t}$ is dependent on the price prediction and its value function of period $t+1$. Recursively, we know storage value at period $V_{t-1}$ is dependent over price forecasts of {future time} periods $t, t+1, t+2,..., T$. In practice, storage opportunity bids have been acknowledged by system operators like California ISO~\citep{caiso_es2}.
\end{remark}

\begin{definition}\label{def:es_bid}\textbf{Storage bid curve.}
   {Storage price bids for discharging and charging are formulated based on the marginal cost function of energy storage, encompassing both physical costs and opportunity costs:}
    \begin{subequations}~\label{eq:esbids_def}
    \begin{align}
        {o}_{t}(p_t) & = \Big[c+\frac{\tau}{\eta}v_{t}(e_{t-1}-\tau p_t/\eta)\Big]^+ \label{eq:bid_pd} \\
        {d}_{t}(b_t) &= \tau \eta v_{t}(e_{t-1} + b_t \tau \eta) \label{eq:bid_pc}
    \end{align}
    \end{subequations}
{where the discharging bid curve~\eqref{eq:bid_pd} is obtained by taking subderivative of discharge cost $C_{o,t}$ with respect to discharge power $p_t$. Similarly, the charging bid curve~\eqref{eq:bid_pc} is obtained by taking subderivative of $C_{d,t}$ with respect to $b_t$. A detailed derivation of equations~\eqref{eq:esbids_def} is given in Appendix~\ref{appendix.bid_derivation}}. $[x]^+ = \max\{0,x\}$ is used to reflect constraint \eqref{eq:c1} that no discharge power would be cleared during negative prices. $v_t(e_t)$ indicates the marginal value function, which is defined as
    \begin{align}
       & v_t(e_t) = \max_h~ \partial V_t(e_t)  \nonumber \\
           &~~ =\max_h~\cap_{x \in \textbf{dom} V_t} \{ h | {V_t}(x) \geq {V_t}(e_t) + h^T(x-e_t) \} \label{eq.defv}
    \end{align}
where $\partial V_t(e_t)$ is the subderivative of $V_t$ at $e_t$. 
{Given the convex function $-V_t(e_t)$, to address scenarios where $V_t$ is piecewise-defined,} we define $v_t(e_t)$ in equation~\eqref{eq.defv} as the maximum value of the subgradient at the breaking points where $V_t$ is not differentiable.
\end{definition}

The storage bid curves~\eqref{eq:esbids_def} show that the bid-offer curves consist of the physical discharge cost $c$ and the opportunity value component based on the marginal value function $v_t$. Note that in cases when markets require price-quantity segment offer curves {in practice}, the bid function can be discretized based on power segments.
Also note that in the RTM, a storage participant is expected to bid as a generator and a load~\citep{sakti2018review}. Consequently, in Definition~\ref{def:es_bid}, the storage bids represented by the functions $o_t(p_t)$ for discharging and $d_t(b_t)$ for charging are formulated based on the {discharging and charging power}, respectively, instead of focusing on the SoC $e_t$. Moreover, when making bids for period $t$, storage has its SoC $e_{t-1}$ at period $t-1$, so the equations in equation~\eqref{eq:esbids_def} are formulated with $e_{t-1}$.

\subsection{Economic Capacity Withholding}

{
Economic capacity withholding~\citep{ferc2020economicwithholding} of energy storage indicates a storage participant strategically chooses not to discharge at a certain time period by submitting bid prices higher than the market clearing price of this time period despite the price being higher enough to recover the storage’s discharge cost.}  
{Consider an example: A battery storage system with a marginal discharge cost of 30 \$/MWh and an efficiency $\eta =1$. charged when the electricity price was low, at 20 \$/MWh. Later, with a forecast price of 60 \$/MWh, this storage can earn profits by discharging. However, it strategically submits a high bid, such as 200 \$/MWh, to effectively avoid discharging. Instead, the storage strategically waits to discharge during a period when the expected future price exceeds 100 \$/MWh. By submitting a reasonable bid during that period, the storage achieves a significantly higher profit margin by its economic capacity withholding behavior.}

{We define the economic capacity withholding in a quantitative way.} 
We establish our framework on the perspective of Harvey and Hogan~\citep{harvey2001market} that states the economic capacity withholding intentions of energy-limited resources can only be identified after-the-fact but must also account for the imperfect foresight. Firstly, we consider the non-withholding capacity bids {as a baseline to establish our definition of economic capacity withholding. Non-withholding bids are characterized as constant bids under the following conditions:} 1) It does not assume market power, 2) it assumes deterministic future price forecast, and 3) the final SoC value is also a linear function. {This baseline serves as a reference for identifying deviations indicative of economic withholding. Secondly, We define energy storage capacity withholding as happens if it submits a bid higher than the non-withholding capacity bids defined in Proposition~\ref{pro:nwb}. }

\begin{proposition}\label{pro:nwb}\textbf{Deterministic storage bids.}
{Consider storage forecast price $\hat \lambda_t$ characterized by its expectation $\mu_t$ and {deviation} $\sigma_t$.}  If $\sigma_t^2=0$ in equation~\eqref{eq:esbid_opti}, then storage bid curves  $o_t(p_t)$ and $d_t(b_t)$ will be constant if $V_{T}(e_{T})$ is a linear function. Hence, {in the RTM clearing problem \eqref{eq:market},} $O_t(p_t)$ and $D_t(b_t)$ as are linear.
\end{proposition}
\begin{proof}
~The proof of this proposition is based on using the Simplex algorithm~\citep{dantzig2003linear}, which states the optimal solution of linear programming must be on a polyhedral vertex, {or at a convex combination of multiple such vertices}. Note that if there is no uncertainty {in storage forecast price $\hat \lambda_t$}, then \eqref{eq:esbid_opti} becomes equivalent to the following linear programming formulation with a strict linear objective function (not piece-wise linear) 
\begin{align}
    \max_{p_t, b_t} \sum_{t=1}^{T} \hat \lambda_{t} (p_{t} - b_{t}) - cp_{t} + V_{T}(e_{T})
\end{align}
subject to \eqref{eq:c1}--\eqref{eq:c4} for all $t\in\mathcal{T}.$
Then $p_t$ and $b_t$ must be {bounded} by one of the constraints, including the upper/lower energy limits and the SoC limits factoring in efficiency. Thus, for all price predictions $\hat \lambda_t$ and beginning SoC $e_{t-1}$ for all time period $t\in\mathcal{T}$, the storage dispatch solution will be either charging at $b_t = \min\{P, (E-e_{t-1})/(\eta \tau) \}$, idle ($p_t=0$, $b_t = 0$), or discharging at $p_t = \min\{P,  e_{t-1}\eta /\tau\}$. Then there is no slope between charge or discharge energy to the price $\hat \lambda_t$. In other words, if we fix all other prices and only perpetuate the price realization of any particular time interval, it will cause the storage at that time period to switch between fully charging, idle, and fully discharging. This is equivalent to a constant supply or bid curve, as the storage will offer full capacity if the price is above or below a certain threshold. Hence, we finished the proof.
\end{proof}

{Given the definition of economic capacity withholding, battery energy storage can assume one of two distinct roles in the market. As an ``honest player," storage seeks to accurately predict future electricity prices and their associated uncertainties, designing bids based on these forecasts without attempting to exert market power. In this role, the storage treats predicted prices and uncertainties as exogenous variables, assuming its actions do not influence market prices or their distribution.
From a market regulator's perspective, an ``honest player'' will either not implement economic capacity withholding, or its withholding will not lead to the exercise of market power.
Conversely, a ``market manipulator'' deviates from this honest behavior by actively seeking to exercise market power, deliberately influencing electricity prices to maximize profits. }

\section{Main Results} \label{sec.main_result}
 {Our proposed framework is capable of analyzing both supply (discharging) and demand (charging) withholding in storage, as these are symmetrical processes with opposite directions in the propagation of value functions with respect to the SoC. For clarity, our theoretical results primarily focus on \textit{supply} (discharging) withholding,} since in electricity markets, as capacity withholdings typically refer to supply offers. {However, the same framework can be readily applied to analyze storage charging (demand) withholding as well.}

{In this section, We first demonstrate the} concavity of storage value functions across all potential price distributions – a crucial property referenced in subsequent findings.
Then we show if the price maintains a consistent expectation but has an \textit{unbounded} distribution, the corresponding storage bid also remains unbounded.
Conversely, with a \textit{bounded} price distribution defined by upper and lower price thresholds, the storage bid will also bounded, and we will detail its upper bound formulation.
We'll further explore various corollaries tied to different price distribution scenarios.
Lastly, we'll comment on the interplay between storage bidding uncertainty models and overall market welfare.

\subsection{{Convex Storage Bids}} \label{subsec.convexbid}
{The system operator is concerned about the convexity of the RTM clearing problem~\eqref{eq:market} and the market power of energy storage, given the awareness that storage participants may engage in economic capacity withholding. Therefore, it is crucial to investigate the convexity of the cost functions of storage bids, $O_t(p_t)$ and $-D_t(b_t)$.  This requires examining the characteristics of the marginal value function $v_t(e_t)$, which decides storage bids $o_t(p_t)$ and $d_t(b_t)$}. 

We now show the value function $V_t(e_t)$ will always be concave for any price distributions given a concave and differentiable end value function $V_T(e_T)$, and the storage bids will always be convex.
\begin{proposition}\label{pro:concave}\textbf{Concave value function.} 
Given a concave end value function $V_{T}(e_T)$,  then $V_{t}(e_{t})$ is concave for all $t\in\mathcal{T}^{'} = \{1,2,3,...,T-1\}$ and for all price distribution functions
$\hat \lambda_t$.
\end{proposition}

Proposition~\ref{pro:concave} states that the value of energy stored is a concave function that follows the law of diminishing returns, hence, the more energy stored, the less marginally valuable they are. We defer the complete proof to Appendix~\ref{appendix.concave}.

\begin{corollary}\label{cor:convexbids} \textbf{Convex storage bids.} 
    Given a concave end value function $V_{T}(e_T)$, ${o}_t(p_t)$  monotonically increases with $p_t$, ${d}_t(b_t)$ monotonically decreases with $b_t$. {Then $O_t(p_t)$ is convex, and $D_t(b_t)$ is concave}, hence the market clearing model in \eqref{eq:market} is always convex.
\end{corollary}

This corollary is trivial to show based on Proposition~\ref{pro:concave}, which proves $V_t(e_t)$ is always concave. Then $v_t(e_t)$ is a monotonically decreasing function. Then we can finish the proofs according to \eqref{eq:esbids_def}.

At this point, we have shown if {the storage value function at the last period $V_{T}(e_T)$ is concave}, then all its bid curves will fit the convex market clearing model under stochastic bid design objectives under any assumed price distributions.

\subsection{Unbounded Economic Capacity Withholding} \label{subsec.ubd}
{In certain deregulated electricity markets, RTM prices can soar to extremely high levels even when day-ahead prices remain in a reasonable range. For instance, during the Texas power crisis in 2019, RTM prices reached approximately 9,000\$/MWh \citep{wikipedia2021texaspowercrisis}, a price later upheld by the state supreme court \citep{txcourts2024casesummaries}. Similar examples have also been observed in Australia \citep{aemo2016highenergypricesa}. These cases underscore the theoretical importance of understanding storage capacity withholding behavior in scenarios where RTM prices exhibit significant or unbounded deviations from DAM prices. However,}
traditional market price analysis focuses solely on the price expectations (i.e., average price values), but there are fewer concerns over the price deviations, hence the volatility. Our results {first show} that even with a given price expectation, the storage economic capacity withholding, i.e., the bid values, can be arbitrarily high based on the price standard deviation when assuming the price follows a Gaussian distribution. 

\begin{theorem}\label{the:ubd}\textbf{Unbounded withholding with Gaussian distributions.} 
Let the concavity of $V_T(e_T)$ hold. Given arbitrary time interval $t \in \mathcal{T}^{'}$, assume the storage participant anticipates future prices $\hat \lambda_{t+1}$ following Gaussian distributions with the fixed expectation $\mu_{t+1}$ but unbounded {deviation} $\sigma_{t+1}$. For a given price expectation over the next time period $\mu_{t+1}$, for an arbitrary bid value $\theta \geq 0$, 
there exists a standard deviation $\sigma_{t+1} \geq \underline{\sigma}$ such that
\begin{align}
v_{t}(0) \geq \theta.~\label{eq.thm1_lbd}
\end{align}
where constant $\underline{\sigma}$ represents the lower bound essential for always ensuring the validity of inequality~\eqref{eq.thm1_lbd} at all time periods.
\end{theorem}

Theorem~\ref{the:ubd} asserts that unbounded price distributions can render storage economic capacity withholding $o_t$ in boundless. This theorem goes beyond the conventional analysis of market prices, highlighting a crucial nuance: Even when storage anticipates a fixed price expectation for the next period, the existence of an unbounded deviation can result in an unbounded marginal value function for storage. Consequently, this leads to unbounded withholding in RTM.

The proof of Theorem~\ref{the:ubd} is based on the following Lemma. {For clarity, consider storage marginal value at SoC level $E_c = \tau P \eta$, which indicates the SoC increment resulting from charging at full power $P$ over a duration of $\tau$.}
\begin{lemma}\label{lem:ubd_pre}
   {Let the concavity of $V_T(e_T)$ hold.} Given $v_T(e_T)\geq 0,~\forall e_T \in [0,~E]$ and the storage's {capacity} $E \geq {2 E_c}$, then $v_t(E_c)\geq 0, \forall t \in \mathcal{T}^{'}$.
\end{lemma}

Lemma~\ref{lem:ubd_pre} asserts that initiating storage charging for a single time period from zero SoC will not immediately result in the marginal value function $v_t(e_t)$ dropping below zero for all $t\in \mathcal{T}$.
The idea to prove Lemma~\ref{lem:ubd_pre} is trivial. {We can find a lower bound of $v_{t}(0)$ and then} we can prove the lower bound larger than zero. Detailed proof can be found in Appendix~\ref{appendix.lemgau}.

Following up on Lemma~\ref{lem:ubd_pre}, the proof to Theorem~\ref{the:ubd} is sketched as follows. First, we must find the exact expression of $\underline{\sigma}$ and prove the unbounded marginal value function $v_t(0)$. The approach to proving this theorem hinges on establishing the monotonic increase of $v_{t}(0)$ concerning the deviation $\sigma_{t+1}$. To accomplish this, we identify the lower bound of $v_{t}(0)$ based on the conclusion of Lemma~\ref{lem:ubd_pre}. {Second, we} demonstrate the monotonic increase of this lower bound with respect to the deviation $\sigma_{t+1}$, thereby providing us with the crucial value of $\underline{\sigma}$.
For detailed proofs of Theorem~\ref{the:ubd}, please refer to the Appendix~\ref{appendix.gau}.

{Theorem~\ref{the:ubd} addresses unbounded withholding under Gaussian distributions. However, in practice, price forecasts may follow a variety of distribution types. We now extend the analysis of storage withholding to consider these different distribution models. }

\begin{corollary}\label{cor:ubd_general}\textbf{Unbounded withholding for distributions {across various classes.}}
{Let the concavity of $V_T(e_T)$ hold.} Given arbitrary time interval $t\in\mathcal{T}^{'}$, for an arbitrary value $\theta \geq 0$ and an arbitrary price expectation $u_{t+1}$, there exists a {distribution} $\hat \lambda_{t+1}$ such that
\begin{align}
    v_{t}(0) \geq \theta \text{ and } \mathbb{E}[\hat \lambda_{t+1}] = \mu_{t+1},
\end{align}
which indicates storage economic capacity withholding, {represented by the price bid $o_t$,} is unbounded.
\end{corollary}

Corollary~\ref{cor:ubd_general} extends the insights of Theorem~\ref{the:ubd} to encompass various distribution types, offering a more comprehensive understanding: In essence, the corollary suggests that storage dynamically updates its opportunity value by effectively aligning it with the latest peak price at its minimum SoC. This implies that the final {unit} of retained energy is strategically earmarked for sale during intervals with the most favorable projected future prices.

Subsequently, the main proof of this corollary articulates that once the storage revises its opportunity value linked to a SoC level to a more elevated price, the prior opportunity value is then transferred to the succeeding elevated SoC levels. This conveys that if storage perceives a unit of energy (SoC) can be traded at a price higher than previously estimated, it logically follows that the subsequent storage unit can be sold at the formerly estimated price. Detailed proofs of Corollary~\ref{cor:ubd_general}, please refer to the Appendix~\ref{appendix.soc}.

\begin{corollary}\label{cor:ubd_interval}\textbf{Unbounded withholding with extended bidding intervals.} 
{Let the concavity of $V_T(e_T)$ hold.} Given arbitrary time {period} $\kappa\in\mathcal{T}$ and starting SoC $e_{\kappa-1}\in[0,E]$, for an arbitrary bid value $\theta \geq 0$ and an arbitrary trajectory of future price expectations $u_t$, there exists a set of future price distributions $\hat \lambda_t$ such that for any $p_{\kappa}\in [0, \min\{P, e_{\kappa-1}\eta/\tau \}]$
\begin{align}
    o_{\kappa}(p_{\kappa}) \geq \theta \text{ and } \mathbb{E}[\hat \lambda_t] = \mu_t | t\in \{\kappa,\dotsc,T\},
\end{align}
if $T-t \geq (e_{\kappa-1}\eta-\tau p_{\kappa})/P$.
\end{corollary}

Corollary~\ref{cor:ubd_interval} shows that given a set of price predictions with fixed price expectations, storage bids can skyrocket to any value within a time interval, as long as the remaining time intervals in the bidding problem are sufficient for the storage to discharge completely. The flexibility for the storage bidder to extend the bidding horizon implies that the storage can strategically withhold any economic capacity based on assumed price distributions without altering the overall expectation. The proof of the Corollary~\ref{cor:ubd_interval} is based on the proofs of Corollary~\ref{cor:ubd_general}, which will be provided in Appendix~\ref{appendix.ubup}.

\subsection{Bounded Economic Capacity Withholding} \label{subsec.bounds}
{The analysis in Subsection~\ref{subsec.ubd} shows that storage unbounded withholding is a result of unbounded RTM prices.}
In practice, the system operator of electricity markets usually imposes price floors and ceilings, hence the distribution of the price is bounded. {Additionally, system operators and regulators want to identify whether a storage participant is exerting market power through economic capacity withholding.} In this section, we present results about the storage economic withholding bid bounds based on bounded prices with fixed or bounded expectations. {These results can serve as a foundation for monitoring and regulating storage economic withholding}.

\begin{theorem}\label{the:bound}\textbf{Bounded withholding with bounded prices.} {Let the concavity of $V_T(e_T)$ hold. } Assume the market prices have an upper bound $\overline{ \lambda}$ and a lower bound $\underline{ \lambda}$. Then, the storage bidding price has the following upper bound given $t \in \mathcal{T}^{'}$:
\begin{subequations}
\begin{align}
    o_t(p_t) &\leq c + \Big(\overline{ \lambda}-c\Big)\sum_{i=t+1}^{T} \alpha_i\prod_{\kappa=t+1}^{i-1} \beta_{\kappa}  \nonumber\\
    &\qquad +  {\prod_{\kappa=t+1}^T \beta_{\kappa}\frac{v_T(0)}{\eta}}\label{eq:the2}
\end{align}
for all possible price distributions satisfying
\begin{align}
    \mathbb{E}(\hat \lambda_{t+1}) = \mu_{t+1} \text{ and } \underline{ \lambda}\leq \hat \lambda_{t+1} \leq \overline{\lambda}
\end{align}
where $
    \alpha_t = ({\mu_t - \underline{\lambda}})/({\overline{\lambda}- \underline{ \lambda}}) $ and $\beta_t = 1-\alpha_t$. {We define  $\prod_{x_1}^{x_2}$ = 1 when $x_1 \geq x_2$}.
\end{subequations}
\end{theorem}

The idea of proving Theorem~\ref{the:bound} is that given the price bounds, then $v_t(0)$ is always maximized if the price is distributed at the maximum and minimum values, hence $\alpha_t$ is the occurrence of $\overline{\lambda}$ during interval $t$, while $\beta_t$ then is the occurrence of $\underline{\lambda}$ while ensuring the price expectation is $\mu_{t+1}$. Detailed proofs are given in Appendix~\ref{appendix.main2}.

{An example illustrating Theorem 2 assumes that storage expects to sell all its stored energy at the minimum market price $\underline{ \lambda}$ at the end of the market period while avoiding discharge at negative prices, as specified in \eqref{eq:c1}. Thus the storage has a linear end-value function $V_T(e_T) = \sigma_t \cdot e_T / \tau$ with $\sigma_t  \geq \Big[(\underline{ \lambda} - c )\eta\Big]^+$. Then we can find the economic withholding bound ${o}_t(p_t) \leq c + \Big(\overline{ \lambda}-c\Big)\sum_{i=t+1}^{T} \alpha_i\prod_{\kappa=t+1}^{i-1} \beta_{\kappa} 
    + \prod_{\kappa=t+1}^T \beta_{\kappa}\frac{\sigma_t}{\tau \eta}$}

The time-varying distribution of the price make Theorem~\ref{the:bound} unapparent to observe {and non-trivial for practical usage}, we now introduce the following corollary which assumes an upper bound of the price expectations.

\begin{corollary}\label{cor:exp}\textbf{Bounded withholding with bounded expectations.}
Following same assumptions in Theorem~\ref{the:bound}, now adding that all price expectations have an upper bound $\mu_t \leq \overline{\mu}$. All assume there is an interval-based discount ratio $\rho\in(0,1]$, then the upper storage bid bound becomes
\begin{subequations}
\begin{align}
    {o}_t(p_t) &\leq c + \frac{\alpha}{1-\rho\beta}(1-\beta^{T-1})(\overline{\lambda}-c) + \beta^T \frac{{v_T(e_T)}}{\eta} \label{cor:exp:eq} \\
    &\leq \frac{\alpha\overline{\lambda} }{1-\rho\beta}+(1-\frac{\alpha}{1-\rho\beta})c\quad(\text{as $T\to\infty$})  \label{cor:exp:eqTinf} 
\end{align}
where scalars $
    \alpha = ({\overline{\mu} - \underline{ \lambda}})/({\overline{ \lambda}- \underline{\lambda}}) $, and $\beta = 1-\alpha$.  
\end{subequations}
\end{corollary}
The derivation of this corollary is trivial by assuming all $\alpha_{t}$ are the same in Theorem~\ref{the:bound}, then the first term in \eqref{eq:the2}, which is a cumulative sum of power ratings, can be reduced to the first term of \eqref{cor:exp:eq}.
Here $\alpha_{t}$ is the upper bound of the occurrence of $\overline{\lambda}$. 

Corollary~\ref{cor:exp} provides a time-invariant upper bound over all bids given bounded price expectations. Note that this upper bound increases with $\overline{\mu}$ and $T$, indicating that given a sufficient long time horizon, the storage bid upper bound will converge to the price upper bound as shown in~\eqref{cor:exp:eqTinf}.

\subsection{Welfare-Aligned Economic Withholding} \label{subsec.welfare}
{Our discussions about unbounded and bounded withholding reveal that energy storage withholding serves as a strategy to manage uncertainty in RTM electricity prices, rather than an act of market manipulation. In this section, we fuehrer examine the role of battery energy storage in electricity markets: Whether the storage acts as an honest participant whose withholding behavior aligns with social welfare. We propose the following proposition:}

\begin{proposition} \label{pro:Lag}
   Assuming the {aggragated} storage capacity is negligibly small,  the variance of the market clearing price $\lambda_t$ scales linearly with the standard variance of the net demand $D_t - w_t$ if the cost function of the {RTM market clearing problem}~\eqref{eq:obj} is quadratic.
\end{proposition}

Proposition~\ref{pro:Lag} states that market clearing price is monotonically associated with net demand, so a higher standard deviation of net demand always results in a higher deviation of clearing prices, {given} the same distribution type. The proof of this proposition is trivial by using the Lagrange function of the RTM {clearing} problem as shown in the Appendix~\ref{appendix.Lag}. {Next, we investigate the relation between storage economic capacity withholding and the uncertainty of net demand in the following corollary}

\begin{corollary} \label{cor:Lag}
    Following same assumptions in Proposition~\ref{pro:Lag}, if the uncertainty models of storage price forecast and systemic net demand follow the same distribution type, the storage utilizing uncertainty models for market bid formulation can potentially lower the overall system cost when the system encounters a suitable level of uncertainty.
\end{corollary}

Corollary~\ref{cor:Lag} is trivial. Given by Proposition~\ref{pro:Lag} that RTM clearing price $\lambda_t$ scales linearly with the net demand, when the storage capacity is small enough that will not impact market clearing outcomes, we can treat $\hat \lambda_t = \lambda_t$, indicating storage has a ``perfect'' forecast of RTM clearing price at period $t$. Therefore, {the solutions of} storage charging/discharging for maximizing own arbitrage profits matches {those} for minimizing system operational costs. 

Corollary~\ref{cor:Lag} reveals that the economic withholding behavior of storage can paradoxically enhance social welfare when the system encounters a suitable level of uncertainty. {This result, together with the analysis of storage convex bids and unbounded/bounded withholding, reveals that storage that strategically designs economic withholding bids may not be a market manipulator but can contribute to social welfare convergence}.

\section{Simulation Results}
\label{sec.simulation}
We demonstrate the effectiveness of the proposed theorems with the following case studies based on the ISO New England test system~\citep{krishnamurthy20158} with an average load of 13~GW.  The transmission constraints are ignored in this study as the New England system is usually not congested~\citep{krishnamurthy20158}.  

The case simulations are conducted using Julia, in which the optimization problems are solved by solver Gurobi. The figures are plotted by Matlab.

\subsection{{Storage Convex Bids}}

This study {shows the concavity of storage value function $V_t(e_t)$ in Proposition~\ref{pro:concave} in Subsection~\ref{subsec.convexbid}, which leads to storage convex bids in Corollary~\ref{cor:convexbids}}.
We analyzes {in practical RTM market how the storage price forecast $\hat \lambda_t$ in Figure~\ref{fig:esv_lmp} shapes the storage marginal valuation $v_t(e_t)$ defined in equation~\eqref{eq:esbid_opti}.}

As shown in Figure~\ref{fig:esv_f}, a key observation is the close connection between the storage price forecast and the resulting marginal values. The peaks and valleys in the price forecast (Figure~\ref{fig:esv_lmp}) directly correspond to the peaks and valleys in the marginal value function (Figure~\ref{fig:esv_f}). This confirms the significant influence of the storage price forecast on storage valuation.

Furthermore, the storage marginal value monotonically decreases as the storage SoC increases. For example, at specific times like 3pm, storage marginal value drops from 45.5 to 11.8 \$/MWh. This aligns with Proposition~\ref{pro:concave}, which suggests that {a storage participant will bid lower prices as their SoC increases, rather than maintaining high bids.} Therefore, storage market power is limited due to individual participants expressing their willingness to charge or discharge through their bidding prices. This finding implies that the system operator can be less concerned about explicitly modeling battery SoC constraints when considering storage bids, as the market itself incentivizes efficient charging and discharging behavior.

\begin{figure}[H]%
	\centering
	
	\subfloat[Price forecast]{
		\includegraphics[trim = 0mm 0mm 0mm 70.5mm, clip, width = .6\columnwidth]{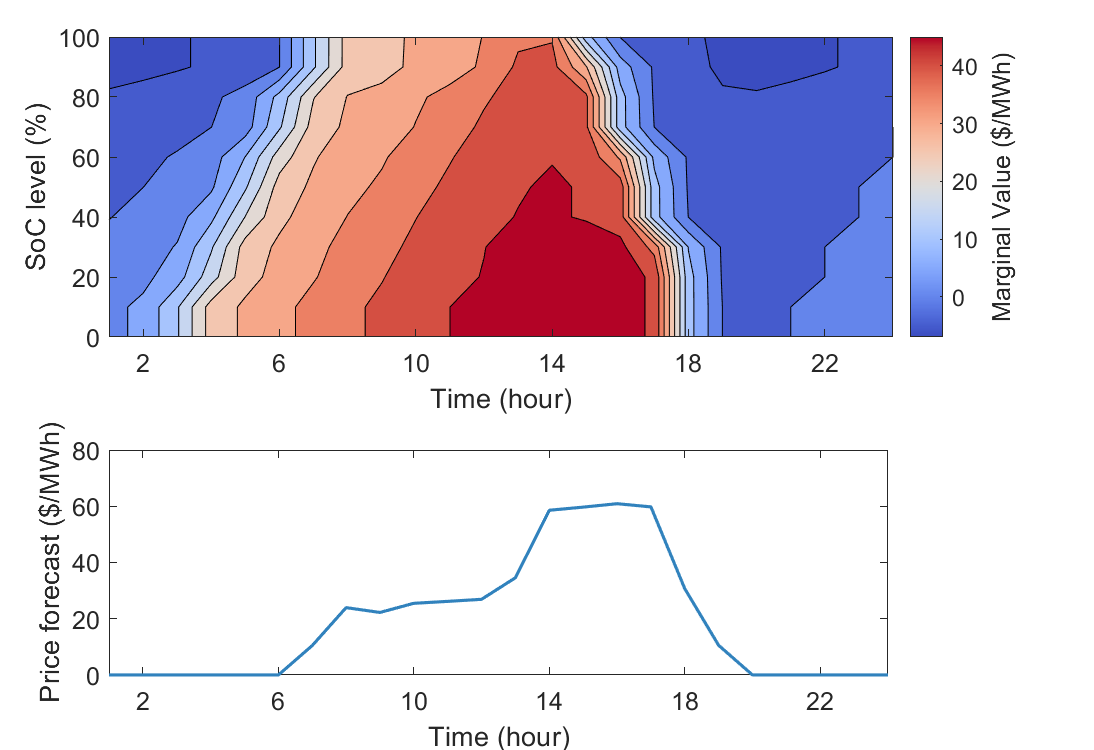}\label{fig:esv_lmp}%
	} \\
	\subfloat[Marginal value]{
		\includegraphics[trim = 0mm 56mm 0mm 0mm, clip, width = .6\columnwidth]{figures/pdbid_2d.png}\label{fig:esv_f}
	}
  \caption{Storage marginal value under price forecast. This case adopts a deterministic price forecast to demonstrate the monotonicity of storage marginal value.  }
    \label{fig:esv}
\end{figure}

\subsection{Storage Unbounded Withholding}~\label{subsec.simulation_ubd}

{This case study validates Theorem~\ref{the:ubd}, demonstrating that storage withholding becomes unbounded when the price forecast follows a Gaussian distribution with a fixed mean and unbounded variance. This price forecast scenario reflects real-world events such as the Texas power crisis in 2019~\citep{wikipedia2021texaspowercrisis} and the Australian RTM price spikes in 2016~\citep{aemo2016highenergypricesa}, where the DAM prices remained reasonable, but RTM prices surged to approximately USD 9,000\$/MWh and AUD 2361.17\$/MWh, respectively - values that can be treated unbounded for theoretical analysis.} 

\begin{figure}[h]%
    \centering
    \includegraphics[trim = 0mm 0mm 0mm 0mm, clip, width = .6\columnwidth]{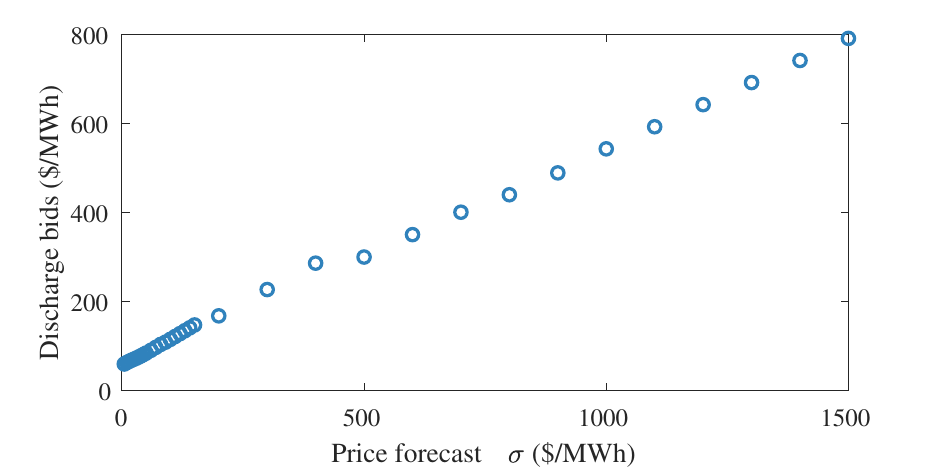}
    \caption{Storage unbounded discharge bids with fixed expectation {of 26.20~\$/MWh} and given deviations from 5 to 1500~\$/MWh. {The marginal discharge cost $c$ is 25~\$/MWh.} }
    \label{fig:esbid_ub}
    \vspace{-1em}
\end{figure}

As shown in Figure~\ref{fig:esbid_ub}, given storage price forecast follows a {Gaussian} distribution with expectation as DAM price and deviation $\sigma$ from 0 to 1500~\$/MWh, {which is a very high price to reflect the two real-world events}. We observe that storage bids are monotonically increasing from {approximately} 10 to 800~\$/MWh {as price deviations grow, even when the forecast price mean is fixed}.  This finding, consistent with Theorem~\ref{the:ubd}, highlights that price deviation plays a crucial role in shaping storage participation in the market. 
{Storage units aim to protect themselves from price fluctuations and the negative effects of uncertainty by withholding, which could inadvertently lead to inefficient market outcomes. 
On the other hand, this monotonic increase in bids demonstrates that the unbounded withholding of storage arises from future price uncertainty rather than an intentional effort to manipulate market power.} To address this issue and better align storage participation with social welfare, we recommend system operators provide more transparent information about price forecasts or set reasonable bounds on their potential fluctuations. This would help storage units make informed decisions and encourage more efficient capacity allocation within the market.

\subsection{Storage Bounded Withholding}~\label{subsec.simulation_bd}
{We validate Theorem~\ref{the:bound}, which states that storage withholding is bounded when system operators impose floors and ceilings on RTM prices. With RTM prices bounded, the storage price forecast is also constrained,} with a lower bound $\underline{ \lambda} =$ 5~\$/MWh and a high bound $\overline{\lambda} =$ 150~\$/MWh. To demonstrate that Theorem~\ref{the:bound} does not require specific price distribution, here we use a uniform distribution whose expectations are the same as Figure~\ref{fig:esv_lmp}, {but the deviations increase from 3 to 120~\$/MWh. When the considered deviation makes forecast price distribution exceed the bounds, we normalize the distribution within the defined bounds.} 

Our results in Figure~\ref{fig:esbid_bdsig} reveal that storage bounds and bids exhibit a dynamic relationship with deviations of the given distribution, in which the bids remain confined within the bounds outlined by equation~\eqref{eq:the2}. We also find that due to the existence of upper and lower bounds, when the deviation is larger than 200~\$/MWh, storage bids will not increase as deviation increases, which {shows the storage economic capacity withholding is also bounded}. Therefore, we advocate for the system operator to establish price bounds that effectively regulate storage bids and bids while simultaneously ensuring societal welfare. {At the same time, the derived price bounds in Subsection~\ref{subsec.bounds} can be used as a baseline for monitoring storage withholding by the system operators.}

\begin{figure}[H]%
	\centering
	\includegraphics[trim = 0mm 0mm 0mm 0mm, clip, width = .5\columnwidth]{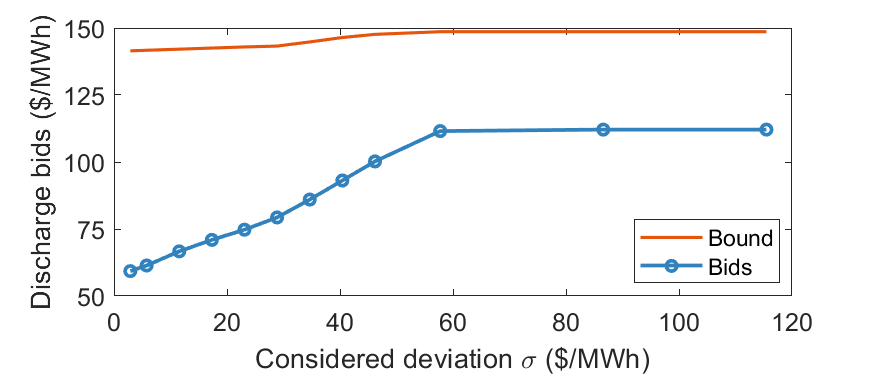}
        \caption{Storage bounded discharge bids with fixed expectations and deviations from 3 to 120~\$/MWh.}
    \label{fig:esbid_bdsig}
    \vspace{-1em}
\end{figure}

\subsection{Welfare-Aligned Economic Withholding } \label{subsec.simulation_welfare}
{In this case, we validate Corollary~\ref{cor:Lag}, which states that storage considering uncertainty in price forecast and withholding improve social welfare. We first examine the ``ideal case" in Subsection~\ref{subsec.results_ideal}, which adheres strictly to the assumption in Corollary~\ref{cor:Lag} that storage capacity is sufficiently small to be negligible.  Next, we present the ``practical case" in Subsection~\ref{subsec.results_practical}, which extends beyond the initial assumption to assess the broader applicability and effectiveness of our findings. This scenario considers a future where the aggregated capacity of energy storage becomes significant and cannot be ignored, a condition anticipated to arise within the next 5–10 years.}

\subsubsection{Ideal case}\label{subsec.results_ideal}
{Here we first validate Proposition~\ref{pro:Lag}, and then validate Corollary~\ref{cor:Lag}. For clarity,} we use a simplified setting with specific characteristics: 
\begin{itemize}
    \item {We modify the ISO-NE test system from Krishnamurthy~\citep{krishnamurthy20158}, by introducing an aggregated} generator with infinite capacity and a quadratic cost function $G_t(g_t) = 10g_t + 0.04(g_t)^2$
    \item {We consider both storage forecast uncertainty and system demand uncertainty follow Gaussian distributions, in which storage price forecasts centered around the DAM price~\citep{tang2016model} with a deviation from 1 to 30 \$/MWh, and the system net demand centered around the DAM demand with a deviation from 0.5 to 3~GW.}
    \item We incorporate a storage capacity of 10 MW/40 MWh, a sufficiently small capacity that will not affect market clearing.
\end{itemize}

{Our first-step result to validate Proposition~\ref{pro:Lag} is shown in Figure~\ref{fig:demand}, which demonstrates} the RTM clearing price monotonically increasing with the standard deviation of the net demand, as predicted by Proposition~\ref{pro:Lag}. Due to the limited storage capacity, the storage price deviation is monotonically increasing with the market clearing price, which  exhibits a linear relationship with the net demand.
\begin{figure}[H]
    \centering
    \includegraphics[trim = 0mm 0mm 0mm 0mm, clip, width = 0.65\columnwidth]{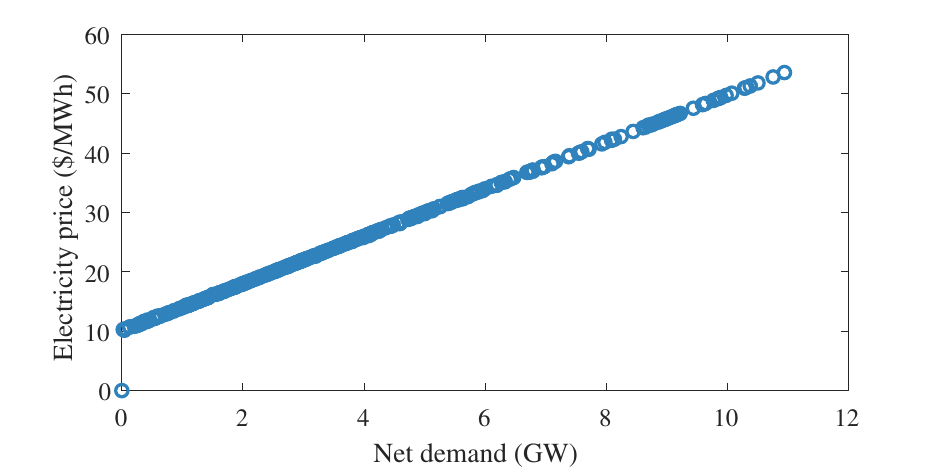}
    \caption{Relationship between market clearing price $\lambda_t$ and net demand.}
    \label{fig:demand}
\end{figure}

{In the second step, we validate Corollary~\ref{cor:Lag} by analyzing the relationship between storage forecast uncertainty and system net load uncertainty. Figure~\ref{fig:cost_ideal}, with net demand deviation on $x$-axis, storage forecast deviation on $y$-axis, and system cost on $z$-axis,} reveals that the minimum cost points approximately lie along the diagonal of the map. This means the storage participant's consideration of price uncertainty can counterpoise the system uncertainties, if the two uncertainties follow the same types of distributions. Note that, Corollary~\ref{cor:Lag} suggests a linear relationship between storage uncertainty and system uncertainty, resulting in strict diagonally aligned points for minimum cost. Figure~\ref{fig:cost_ideal} exhibits slight differences from this ideal scenario. These differences can be attributed to 1) the discrete intervals used for both the $x$-axis (2~\$/MWh) and $y$-axis (0.5~GW), 2) the limited number of Monte-Carlo realizations, and 3) the time-dependent dispatch characteristic of RTM. Notably, Corollary~\ref{cor:Lag} focuses on a single period $t$, while simulation in this case involves storage's SoC and charging/discharging decisions across multiple periods, influenced by the entire price prediction curve.
\begin{figure}[H]%
	\centering
	\subfloat[System cost]{
		\includegraphics[trim = 0mm 0mm 0mm 0mm, clip, width = .5\columnwidth]{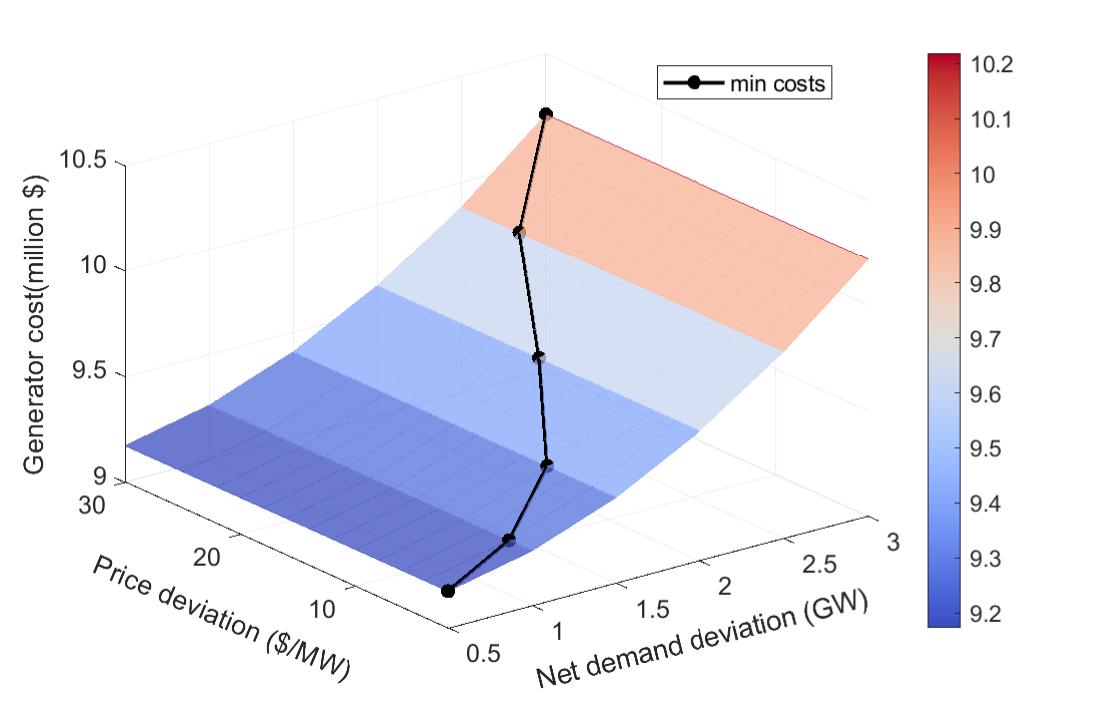}\label{fig:cost_ideal}%
	}
	\subfloat[Storage profit]{
		\includegraphics[trim = 0mm 0mm 0mm 0mm, clip, width = .5\columnwidth]{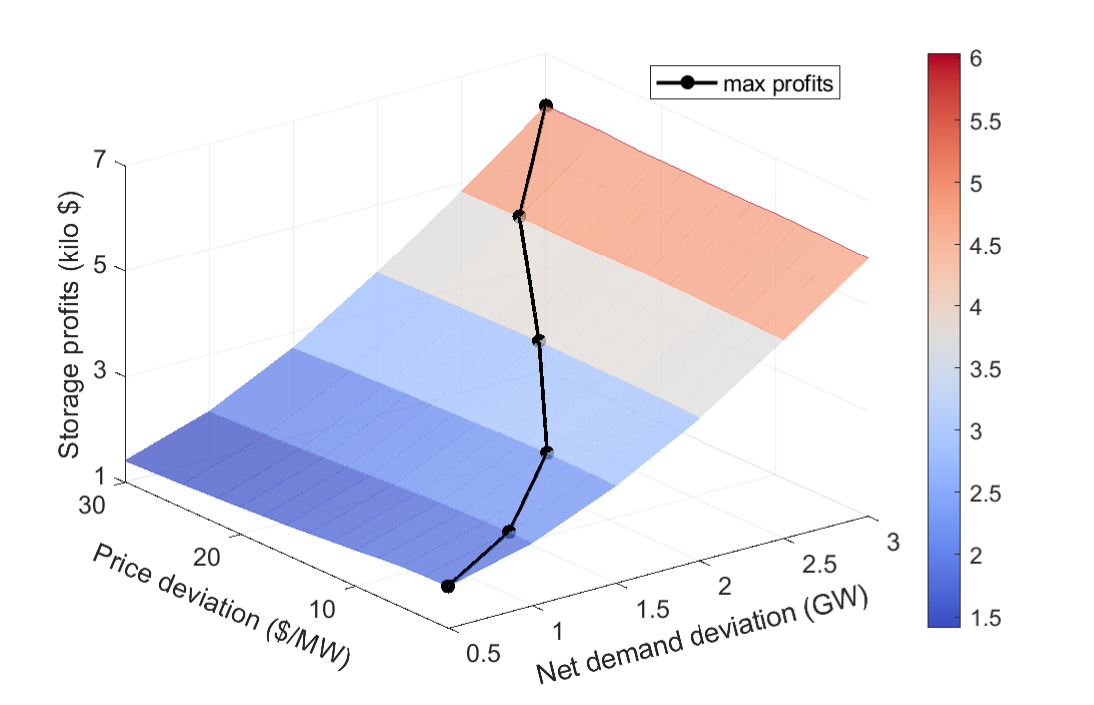}\label{fig:profit_ideal}
	}
  \caption{System cost and profits in the RTM markets under the ideal case.}
    \label{fig:ideal}
    \vspace{-1em}
\end{figure}

Moreover, in this case, due to limited storage capacity, storage participants possess near-perfect predictions of the market price. {This results in the minimum cost in Figure~\ref{fig:cost_ideal} and maximum profit in Figure~\ref{fig:profit_ideal} points coinciding, indicating alignment between storage economic withholding of maximizing profits and social welfare of minimizing system costs.}

\subsubsection{Practical case}\label{subsec.results_practical}

{The real-world RTM differs from the idealized cases described in Subsection~\ref{subsec.results_ideal} in several key ways. First, the market includes a large number of generators with different cost functions. Second, the aggregated capacity of individual storage participants will become significant within the next 5–10 years, due to increasing storage and renewable integration into the grid. Third, wind generation uncertainty depends on real-time wind output rather than a fixed value. For instance, when daily wind generation ranges from 0.1 GW to 10 GW, it is not realistic to assume a constant deviation of 0.5 GW across all time periods, as was done in Subsection~\ref{subsec.results_ideal}}. 

In this section, we bridge the gap between {future} real-world electricity market and the ideal cases in Section~\ref{subsec.results_ideal}, and study how storage withholding impacts social welfare under real-world market settings as mentioned above. We employ the ISO-NE test system with 76 generators and consider a storage capacity of 2.5~GW, which accounts for around 20\% of the average load. For better representatives, we use the K-means approach to generate five representative demand and wind scenarios from the one-year data, in which the wind capacity is 26~GW with an average wind capacity factor of 0.4. We consider the wind uncertainty characterized by normal distribution whose expectation is DA wind generation, and deviation is proportional to DAM wind generation as shown by the $x$-axis in Figures~\ref{fig:real}.

Figures~\ref{fig:cost_real} and \ref{fig:profit_real} show the simulated cost and profit outcomes, respectively. Notably, the locations of minimum cost and maximum profit no longer coincide, although their trends remain aligned. This divergence arises from the clearing price deviating from forecasts due to three factors: increased storage capacity, multiple wind scenarios, and wind generation fluctuations.

Increasing storage capacity amplifies the impact of its charging and discharging actions on market clearing, leading to deviations from price forecasts. The multiple scenarios, representing low-wind and high-wind days, further contribute to uncertainty. Unlike the idealized case in Figure~\ref{fig:ideal}, where uncertainty directly maps to absolute deviation, the uncertainty on the $x$-axis in Figure~\ref{fig:real} reflects a combination of wind generation variation and scenario weights.

On the positive side, Figure~\ref{fig:real} reveals that the trend of the storage price forecast deviation aligns with the system cost and storage profit, suggesting that economic withholding can generally coincide with social welfare. However, this alignment is not guaranteed in all cases.

\begin{figure}[H]%
	\centering
	
	\subfloat[System cost]{
		\includegraphics[trim = 0mm 0mm 0mm 0mm, clip, width = .5\columnwidth]{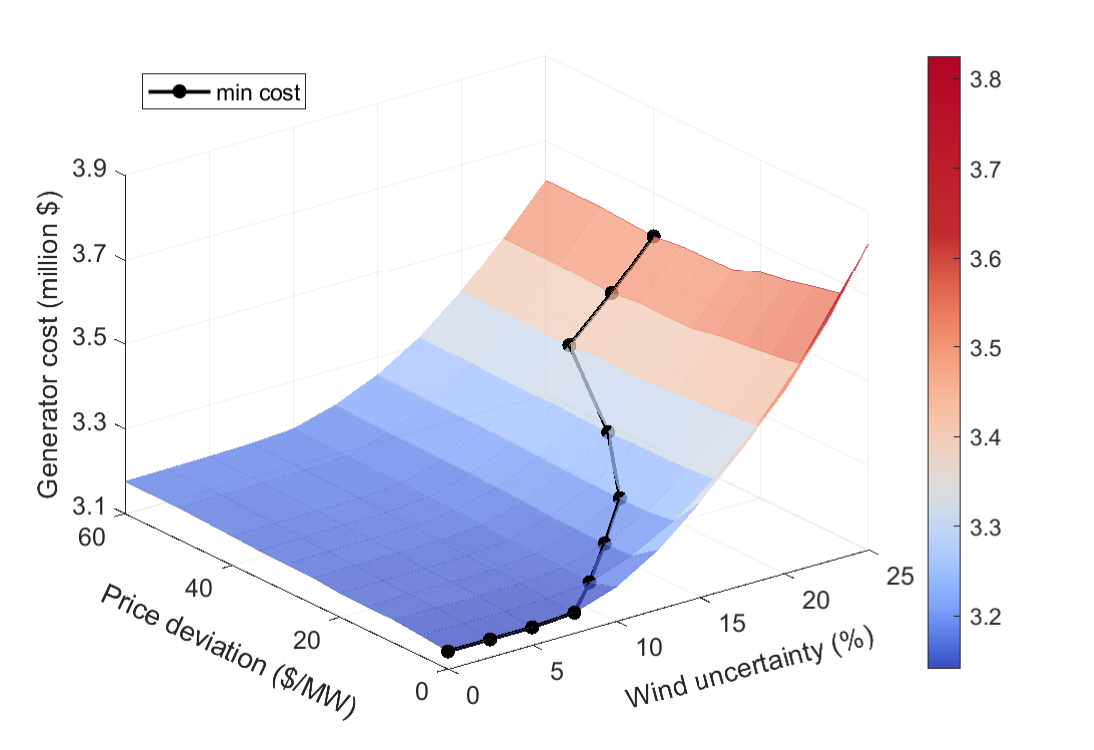}\label{fig:cost_real}%
	}
	\subfloat[Storage profit]{
		\includegraphics[trim = 0mm 0mm 0mm 0mm, clip, width = .5\columnwidth]{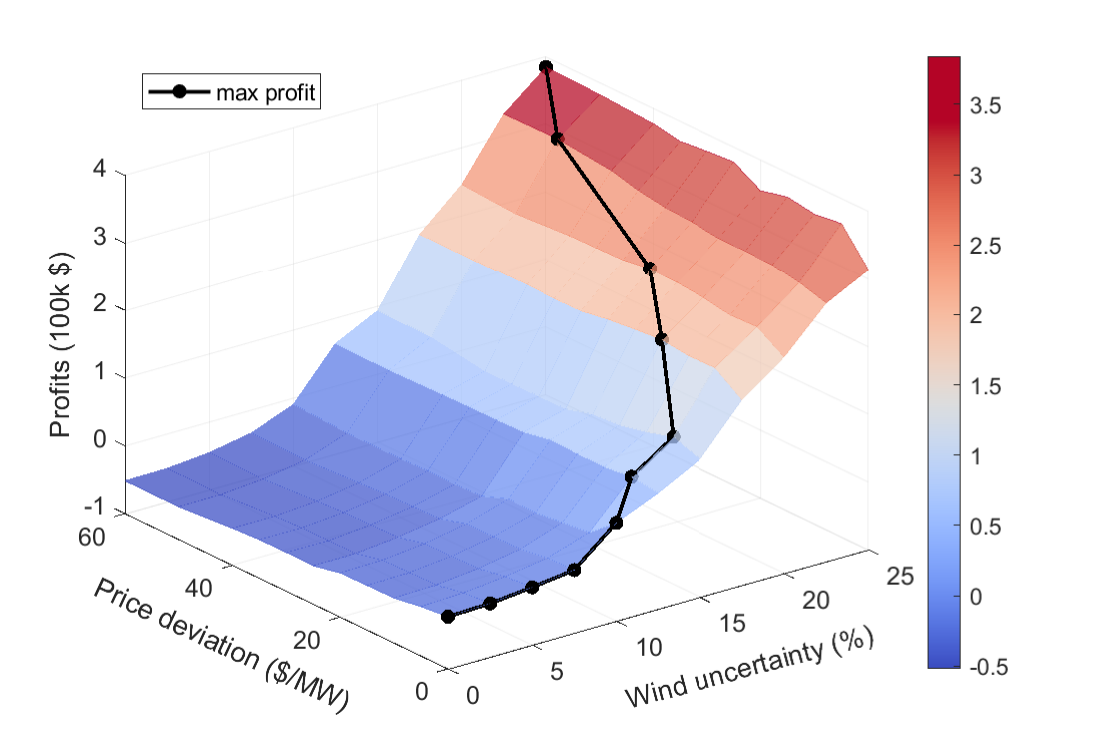}\label{fig:profit_real}
	}
  \caption{Storage participation in markets under practical cases. This setup uses practical system settings of ISO-NE with 76 generators. Storage capacity accounts for 20\% load, having a significant impact on market clearing.  The results come from 5 representative demand and wind scenarios from the one-year data.}
    \label{fig:real}
\end{figure}

\section{Conclusion}
{In this study, we introduce a novel theoretical framework to incorporate uncertainty models in price bidding design for storage participating in competitive electricity markets. }
We showed that a monotonic concave function that defines storage valuation, which reveals storage exercising of market power is limited because a higher SoC level will result in lower bidding prices. 
We also found that storage considering larger price uncertainty will result in higher storage capacity withholding, {and gave the withholding bounds if the wholesale markets have bounds on electricity prices.} We demonstrate that uncertainty introduced into storage bidding could counterpoise the system uncertainties, the effect of this counterpoise is observed to impact the dispatch of all the distributed generators participating in the market and could potentially minimize the system costs and storage profits. 

{Based on the above findings, we conclude that} energy storage is not a market manipulator but an honest player contributing to social welfare, if the system operator can share the information about renewable generation uncertainty. Crucially, we reveal that under certain conditions, the economic withholding behavior of storage can paradoxically enhance social welfare, especially when price models account for systemic uncertainties. This counterintuitive finding underscores the complexity of storage behavior in market dynamics. To substantiate our theoretical findings, we employ numerical simulations, including a case study modeled on the ISO New England system. These simulations provide empirical evidence supporting our theoretical predictions and offer insights into the practical implications of storage strategies in real-world electricity markets. {Future work could investigate the market equilibrium of storage participation, the design of market policies to incentivize storage to counteract system uncertainties, and the impact of advanced prediction methods on storage withholding behavior.}

\bibliography{main}		

\newpage
\section*{Appendix (Online)}
\subsection{Derivation of Definition~\ref{def:es_bid}: Storage Bid Curves}\label{appendix.bid_derivation}
{
We here derive the storage bidding function in~\eqref{eq:esbids_def}. 
Storage participant is required to submit its discharging bid $o_t(p_t)$ and charging price $d_t(b_t)$ bids to the system operator, which is the marginal cost of cost functions $C_{o,t}$ and $C_{d,t}$, respectively. }

{Take the derivation of discharge price bid for example. We first consider the case that $V_t(e_t)$ is differentiable. Then we utilize \eqref{eq:c3} to calculate the partial derivative of $e_t$ to $p_t$ and $b_t$ with partial derivative chain rule. }
    \begin{subequations}~\label{eq:bb1}
    \begin{align}
        {o}_{t}(p_t) &=  \Big[ \frac{\partial C_{o,t}}{\partial p_t}\Big]^+ = \Big[\pdv{c p_t}{p_t} - \pdv{V_t}{p_t}\Big]^+ \nonumber\\
        &= \Big[c - \pdv{V_t}{e_t} \pdv{e_t}{p_t}\Big]^+ = \Big[c + \frac{\tau}{\eta}\pdv{V_t}{e_t} \Big]^+ \nonumber\\
        &= \Big[c+\frac{\tau}{\eta}v_{t}(e_{t-1}-p_t/\eta)\Big]^+ \label{eq:bid_pd_deri_og} 
    \end{align}
   {We then onsider the case of piecewise-defined $V_t(e_t)$ with breaking points, and utilize \eqref{eq:c3} $e_{t} - e_{t-1} = (-p_{t}/\eta + b_{t}\eta) \tau$ again to substitute $e_t$ with $e_{t-1}$ and discharge power since $e_{t-1}$ is a known state to the storage operator and serves as input to the bid curve. We obtain}
        \begin{align}
        {o}_{t}(p_t) &= \Big[c+\frac{\tau}{\eta}v_{t}(e_{t})\Big]^+ \nonumber\\
        &= \Big[c+\frac{\tau}{\eta}v_{t}(e_{t-1}-p_t/\eta)\Big]^+ \label{eq:bid_pd_deri} 
    \end{align}    
where the range of the bid curves follows the feasible range of discharge and charge energy as defined in \eqref{eq:obj-c1} and \eqref{eq:obj-c2}, respectively. 
Similarly, for charging bids, we have
        \begin{align}
        {d}_{t}(b_t) &=  \Big[ \partial C_{d,t} \Big]^+ 
        = 0+ \tau \eta v_{t}(e_{t-1}+b_t \eta). \label{eq:bid_pc_deri} 
    \end{align}
    \end{subequations}

\subsection{Complete market simulation formulation} \label{appendix.market}
Here we present the model of a typical two-stage energy market consisting of  the DAM for unit commitment and the RTM for sequential real-time dispatches.
Our study focuses on energy storage participation in RTM.

\subsubsection{Optimization model of day-ahead unit commitment} \label{appendix.uc}
This subsection provides a detailed description of the day-ahead dispatch model for unit commitment.

\begin{subequations}\label{eq2}
Based on the assumption {in Subsection~\ref{subsec.prelim}}, the objective function of day-ahead unit commitment minimizes daily generation costs of conventional generators and the {physical} discharge costs of energy storage {who participates in the DAM}:

\begin{align}
    \min_{g_{i,t},u_{i,t},y_{i,t},p_{t}} &\sum_{t=1}^{T}  \sum_{i=1}^{N_g} [G_{i,t}(g_{i,t}) + c^n_i \cdot u_{i,t} + c^s_i \cdot y_{i,t}] \nonumber \\
    & + \sum_{t=1}^{T} {\sum_{j=1}^{N_{e}^{DA}}} c \cdot p_{t}  
\end{align}
where $g_{i,t}$ is the electric power generation of conventional generator $i$ at period $t$. $u_{i,t}$ is a binary variable indicating whether generator $i$ is on at period $t$, and $y_{i,t}$ is a binary variable indicating whether generator $i$ turns on at period $t$. 
$c^n_i$ is the generator no load cost, $c^s_i$ denotes generator start-up cost. $p_{t}$ is the storage discharging energy of participant $j$ at period $t$. $N_g$ indicates the number of conventional generators. $N_{e}^{DA}$ indicates the number of storage participants in the DAM.

The constraints are:

Generator minimum and maximum limits 
\begin{align}
    \mathrm{Gmin}_i \cdot u_{i,t} & \leq g_{i,t} \leq \mathrm{Gmax}_i \cdot u_{i,t}, \label{glimit} \ 
\end{align}
where $\mathrm{Gmin}_i$ and $\mathrm{Gmax}_i$ denote the minimum and maximum generation of conventional generator $i$. 

Generator ramping constraints 
\begin{align}
    -\mathrm{RR}_i & \leq g_{i,t}-g_{i,t-1} \leq \mathrm{RR}_i + \mathrm{Gmin}_i \cdot y_{i,t} \label{gramp}\
\end{align}
where $\mathrm{RR}_i$ is the ramp rate of generator $i$. 

Generator start-up and shut-down logic constraints
\begin{align}
    y_{i,t} - z_{i,t} &= u_{i,t} - u_{i,t-1}   \\
    y_{i,t} + z_{i,t} &\leq 1
\end{align}
where $z_{i,t}$ is a binary variable indicating whether generator $i$ turns off at period $t$.

To address the real-time fluctuations caused by renewable generation, the {DAM} is required to have synchronous reserve capacity provided by conventional generators:
\begin{align}
    \sum_{i=1}^{N_g} r_{i,t} &\geq (20\%)w_t \\
    r_{i,t} &\leq \mathrm{Gmax}_i \cdot u_{i,t} - g_{i,t} \\
    r_{i,t} &\leq RR_i
\end{align}
where $w_{t}$ is accommodated wind generation during time period $t$, and $r_{i,t}$ is the reserve capacity.

Generator minimum up time constraint
\begin{align}
    \sum_{\kappa = \max\{t - \mathrm{Tup}_i + 1, 1\}}^t y_{i,\kappa} &\leq u_{i,t}   \\
    \sum_{\kappa = \max\{t - \mathrm{Tdn}_i + 1,1\}}^t z_{i,\kappa} &\leq 1-u_{i,t} 
\end{align}
where $\mathrm{Tup}_i$ and $\mathrm{Tdn}_i$ are the maximum up time and minimum down time of generator $i$, respectively.

Storage {power} constraints
\begin{align}
    &0 \leq b_{t} \leq P \label{pu_c1} \\
    &0\leq p_{t} \leq P  \label{pu_c1p}\\ 
    &\text{$b_{t}$ or $p_{t}$ is zero for any $t \in\{1,2,\dotsc, T\}$ } \label{pu_c2}
\end{align}
where \eqref{pu_c2} enforces that storage cannot charge and discharge during the same time period.

Storage SoC limits
\begin{align}
    & e_{t} - e_{t-1} = \tau  (-p_{t} /\eta + b_{t} \eta) \label{pu_c3}  \\
    & 0 \leq e_{t}  \leq E \label{pu_c4}  \
\end{align}
The electric power balance between the generation side and the load side
\begin{align}
    \sum_{i=1}^{N_g} g_{i,t} + w_t + \sum_{j=1}^{N_e} p_{t} = \hat L_t +  \sum_{j=1}^{N_e} b_{t} \label{eq:da_balance}
\end{align}
where 
the dual variable $\lambda^{DA}_t$ associated with constraint~\eqref{eq:da_balance} is the day-ahead electricity price at period $t$.

The wind generation limits
\begin{align}
    0 \leq w_t \leq \hat {W}_{t} \label{wind} 
\end{align}
where $\hat{W}_{t}$ denotes the day-ahead wind generation forecast at period $t$.

\end{subequations}

Note that in the main text equation~\eqref{eq:market}, for clarity, we use $g_{t} = \sum_{i=1}^{N_g} g_{i,t}$ to denote the aggregated generation power. Then $G_{t}$ indicates the aggregated generation costs.

\subsubsection{Real-time market clearing} \label{appendix.rt} 
\begin{subequations}\label{eq3}
The real-time dispatch aims to minimize the generation costs of generators and the opportunity costs of storage at time period $t$: 
\begin{align}
    &\min \sum_{i=1}^{N_g} G_{t,i}(g_{t,i}) +  {\sum_{j=1}^{N_e} O_t (p_t) - \sum_{j=1}^{N_e} D_t (b_t)}
\end{align}
{where cost functions $O_t(p_t) = \int o_t \,dp_t$ and $D_t(b_t)= \int d_t \,db_t$. In practice, system operators might require a piece-wise linear discharge and charge bids~\citep{caiso2021energystorage}, which indicates the system operators may approximate $O_t(p_t) \approx \sum_k o_{t,k} p_{t,k}, D_t(b_t) \approx \sum_k d_{t,k} b_{t,k}$, in which $o_{t,k}$ and $d_{t,k}$ are discharge and charge bids at segment $k$. $p_{t,k}$ and $b_{t,k}$ are discharge and charge power at period $t$ from segment $k$. Then $p_t = \sum_k p_{t,k}$ and $b_t = \sum_k b_{t,k}$. }

The remaining constraints are analogous to those in the day-ahead unit commitment model, but they are applied for each time period $t=1$, $t=2$, ..., $t=T$ in the RTM dispatch:
\begin{align}
    & \mathrm{Gmin}_i \cdot u_{i,t} \leq g_{i,t} \leq \mathrm{Gmax}_i \cdot u_{i,t} \label{glimit1} \\
    & -\mathrm{RR}_i\leq g_{i,t}-g_{i,t-1} \leq \mathrm{RR}_i + \mathrm{Gmin}_i \cdot y_{i,t}  \label{gramp1} \\
    & \sum_{i} g_{i,t} + w_t +   \sum_{\mathcal{S}} p_{t} = L_t   - \sum_{\mathcal{S}} b_{t} \label{gbalance1} \\
    & 0 \leq w_t \leq W_t \label{wind1} \\
    & 0 \leq p_{t} \leq P \label{pr_c1} \\
    & 0 \leq b_{t} \leq P \label{pr_c1p}\\ 
    & \text{$p_{t}$ or $b_{t}$ is zero at $t$ } \label{pr_c2} \\
    & e_{t} - e_{t-1} = \tau(-p_{t}/\eta + b_{t}\eta) \label{pr_c3} \\
    & 0 \leq e_{t} \leq E \label{pr_c4} \
\end{align}
where $W_t$ is the real-time wind {capacity}, which has deviations from the day-ahead forecast $\tilde{W}_{t}$.

Constraints \eqref{glimit1} and \eqref{gramp1} set the power generation limits and the ramping limits of conventional generators, respectively. Note here $u_{i,t}$ and $y_{i,t}$ are known values from {DAM unit commitment results}. Therefore, the real-time dispatch is a standard convex optimization problem with quadratic objective functions and linear constraints, which is efficient to solve. Constraint \eqref{gbalance1} is the power balance constraint between the supply side and the demand side. Constraint \eqref{wind1} indicates the wind generation should not exceed the wind capacity. Constraints \eqref{pr_c1} and \eqref{pr_c1p} specify the limits on storage charging and discharging power, respectively, and constraint \eqref{pr_c2} ensures that storage cannot charge and discharge simultaneously during the same time period. Equation \eqref{pr_c3} defines the relation between SoC level and charging/discharging power, where $e_{t-1}$ is a known value, because results and states at $1,2,...,t-1$ have been obtained when executing real-time dispatch at period $t$. Constraint \eqref{pr_c4} is the SoC limit at period $t$.

\end{subequations}

\subsection{Proof of Proposition~\ref{pro:concave}} \label{appendix.concave}

{The proof of Proposition~\ref{pro:concave} starts from period $T-1$. Given a concave and differentiable end value function $V_{T}(e_{T})$, function $Q_{T-1}(e_{T-1})$ is a continuous and piecewise function, and is differentiable apart from some isolated points (called breaking points) at which the inequality constraint of problem~\eqref{eq:esbid_opti} shifts between being binding and non-binding. We show function $Q_{T-1}(e_{T-1})$ is concave at the intervals over which $Q_{T-1}(e_{T-1})$ is differentiable. We then demonstrate the concavity of $Q_{T-1}(e_{T-1})$ still holds at breaking points, which subsequently implies the concavity of $V_{T-1}(e_{T-1})$. Finally, by working recursively we can prove the concavity of $V_{t}(e_{t})$ for $t = T-2, T-3,... 1$, thereby concluding the proof. }

At period $T-1$, assume that there exists an optimal solution of problem~\eqref{eq:esbid_opti}, $p_{T-1}^*$ and $b_{T-1}^*$, then consider the Karush–Kuhn–Tucker (KKT) condition for charging and discharging cases, respectively: 
\begin{subequations}~\label{eq.pro1_kkt}
\begin{align}
    -\hat \lambda_{T} + \pdv{V_{T}(e_{T})}{b_T^*}  + \overline \mu_e  + \overline \mu_b - \underline \mu_b &= 0  \\
    \hat \lambda_{T} -c + \pdv{V_{T}(e_{T})}{p_T^*}  + \underline \mu_e  + \overline \mu_p - \underline \mu_p &= 0
\end{align}
\end{subequations}
where $\underline \mu_e$ and $\overline \mu_e$ is the dual variables associated with inequality constraint~\eqref{eq:c4}, where the underline and overline dual variables correspond to the lower limit and upper limit, respectively. $\bar \mu_b$  and $\underline \mu_b$ are the dual variables of inequality constraint~\eqref{eq:c2}. $\bar \mu_p$ and $\underline \mu_p$ are the dual variables of inequality constraint~\eqref{eq:c1}. At equilibrium, using the complementary slackness conditions gives $\underline \mu_b = 0$ and $\underline \mu_p = 0$. 

We now analyze the maximized cost $Q_{T-1}(e_{T-1})$. The optimal solution $p_{T-1}^*$ and $b_{T-1}^*$ and the maximized cost $Q_{T-1}$ are dependent on price forecast $\hat \lambda_{T}$ and storage SoC $e_{T-1}$. {Since the objective function $Q_{T-1}(e_{T-1})$ is continuous and $e_{T-1}$ has a linear relationship with $p_{T-1}^*$ and $b_{T-1}^*$  a small change in $e_{T-1}$ will result in correspondingly small changes in the optimal solution. }
At the differentiable intervals, consider the partial differential of $Q_{T-1}$ with respect to $e_{T-1}$:
 \begin{align} \label{eq.pro1_dQ}
    \pdv{Q_{T-1}}{e_{T-1}} &= \hat \lambda_{T}(\pdv{p_{T}^*}{e_{T-1}} - \pdv{b_{T}^*}{e_{T-1}})          - c\pdv{p_{T}^*}{e_{T-1}} \\
             &~~~~+ \pdv{V_{T}(e_{T})}{e_{T-1}} \nonumber
\end{align}

{
To investigate the result of~\eqref{eq.pro1_dQ}, consider one scenario that storage reaches the capacity limit such that constraint~\eqref{eq:c4} is binding, which results in $e_T = E$ or 0. Thus, $b_T^* = \frac{E - e_{T-1}}{\tau \eta}$ or $p_T^* = e_{T-1}\eta /\tau$ becomes a constant. We then get $\pdv{Q_{T-1}}{e_{T-1}} = 0$, indicating the second order derivative of $Q_{T-1}(e_{T-1})$ is zero. Thus, $Q_{T-1}(e_{T-1})$ is concave in this scenario.
}

{
Consider another scenario that constraint~\eqref{eq:c4} is not binding, which indicates storage does not reach capacity limits, then $\bar \mu_e = 0$ for charging and $\underline \mu_e = 0$ for discharging. We consider five cases: In case 1, storage is fully discharging at $p_T^* = P$; In case 2, storage is partially charging where $0<p_T^* < P$; In case 3, storage is neither charging nor discharging, where $p_T^* = b_T^* = 0$; In case 4, storage is partially charging  $0<b_T^* < P$; In case 5, storage is fully charging at $b_T^* = P$. }
In case 1 and 5, $p_T^*$ and $b_T^*$ are constants, then we obtain $\pdv{e_{T}}{e_{T-1}} = 1$ from equality constrain~\eqref{eq:c3}.
In cases 2 and 4, from complementary slackness we know $\bar \mu_b = 0$ and $\bar u_p = 0$, respectively, which gives the result of $\pdv{e_{T}}{e_{T-1}} = 0$ from equality constrain~\eqref{eq:c3}. Then $\pdv{p_T^*}{e_{T-1}}=\eta/\tau$ in case 2 and $\pdv{b_T^*}{e_{T-1}}= -1/(\tau \eta)$ in case 4, and $\pdv{e_{T}}{e_{T-1}} = 1$ for both cases. For clarity, the result is summarized in equation~\eqref{eq.pro1_bind}.
\begin{subequations} \label{eq.pro1_bind}
 \begin{align}
\pdv{p_{T}^*}{e_{T-1}} &=
    \begin{cases}
    0; ~~\eqref{eq:c1} \text{ is binding}  \label{eq:v4a}
    \\
    \eta/\tau;~~ \text{Otherwise}
     \end{cases}
     \\
\pdv{b_{T}^*}{e_{T-1}} &=
     \begin{cases}
    0; ~~\eqref{eq:c2} \text{ is binding}  \label{eq:v4b}
    \\
    -1/(\tau\eta );~~ \text{Otherwise}
     \end{cases}
     \\
\pdv{e_{T}}{e_{T-1}}  &=
     \begin{cases}
    1; ~~\eqref{eq:c1}~\text{or } \eqref{eq:c2}~\text{is binding} \label{eq:v4c}
    \\
    0;~~ \text{Otherwise}
     \end{cases}
\end{align}
\end{subequations}

Now calculate the value of $\pdv{Q_{T-1}}{e_{T-1}}$ under the five cases:
\begin{subequations}
\par
\begin{itemize}
\item Case 1: $p_{T}^*=P$ and $b_{T}^*=0$. In this case, the battery is discharging at maximum power. The constraint \eqref{eq:c1} is binding. From result~\eqref{eq.pro1_bind}, using $\pdv{p_{T}^*}{e_{T-1}} = 0$ and $\pdv{b_{T}}{e_{T-1}} = 0$ gives
\begin{align}
   \pdv{Q_{T-1}}{e_{T-1}} = \pdv{V_{T}}{e_{T}}
    \label{eq:q1}
\end{align}

\item Case 2: $0 < p_{T}^* < P$ and $b_{T}^*=0$. In this case, the battery is partially discharging between the lower and upper power boundaries. Constraint \eqref{eq:c1} is not binding. Using result~\eqref{eq.pro1_bind}  gives
\begin{align}
\pdv{Q_{T-1}}{e_{T-1}} = \frac{\eta}{\tau}(\hat \lambda_{T}- c)
 \label{eq:q2}
\end{align}

\item Case 3: $p_{T}^*=0$ and $b_{T}^*=0$. In this case, the battery is neither charging nor discharging. Therefore, 
\begin{align}
  \pdv{Q_{T-1}}{e_{T-1}} = \pdv{V_{T}}{e_{T}}
 \label{eq:q5}
\end{align}

\item Case 4: $p_{T}^*=0$ and $0<b_{T}^* < P$. In this case, the battery is only charging between the lower and upper power boundaries, i.e, $0 < b_{T}^* < P$. Constraint \eqref{eq:c2} is not binding. Consequently, 
\begin{align}
  \pdv{Q_{T-1}}{e_{T-1}} = \frac{1}{\tau \eta} \hat \lambda_{T}
 \label{eq:q4}
\end{align}

\item Case 5: $p_{T}^*=0$ and $b_{T}^*=P$. This case is similar to case 1, where $b_{T}^*$ is at the upper boundary, and constraint \eqref{eq:c2} is binding. Therefore, we have a similar result as case 1.
\begin{align}
  \pdv{Q_{T-1}}{e_{T-1}} = \pdv{V_{T}}{e_{T}}
 \label{eq:q3}
\end{align}

\end{itemize}
\end{subequations}

Then we calculate the second-order derivative of $Q_{T-1}(e_{T-1})$ based on the results of~\eqref{eq:q1}-\eqref{eq:q5}:
 \begin{align}
  {\partial^2 Q_{T-1}(e_{T-1})}=
    \begin{cases}
    \{0 \}; ~~ \text{ Case 2, 4} 
    \\
     \frac{\partial^2 V_{T}}{\partial^2 e_{T}} ; ~~ \text{  Case 1, 3, and 5}
     \end{cases}
 \end{align}
Summarizing the above results, given that $V_{T}(e_T)$ is concave and differentiable and $Q_{T-1}(e_{T-1})$ is in the differentiable intervals, we have $\frac{\partial^2 V_{T}}{\partial^2 e_{T}} \leq 0$. Therefore, $Q_{T-1}(e_{T-1})$ is concave under all given $\hat \lambda_t$.

Now we extend our result to the breaking points of $Q_{T-1}(e_{T-1})$. 
At the breaking points, consider the subderivative of $Q_{T-1}$ with respect to $e_{T-1}$:
\begin{align}
    &\partial{Q_{T-1}}(e_{T-1}) = \cap_{z \in \textbf{dom} Q_{T-1}} \{ g | {Q_{T-1}}(z)\nonumber \\
    &~~~~~~~~~~~~~~~~~~~~~~~~~~~\geq {Q_{T-1}}(e_{T-1}) + g^T(z-e_{T-1}) \}
\end{align}
where $\partial{Q_{T-1}}(e_{T-1})$ is a closed set. From the analysis above, $Q_{T-1}(e_{T-1})$ is concave at both sides of every breaking point. Thus, the second-order subderivative of $Q_{T-1}$, i.e., $\partial^2{Q_{T-1}}(e_{T-1})$, is a set with all elements less than or equal to zero, which means $Q_{T-1}(e_{T-1})$ is a concave function. {Since $V_{T-1}(e_{T-1})$ is the conditional expectation of $Q_{T-1}(e_{T-1}|\hat \lambda_{T-1})$ over $\hat \lambda_{T-1}$, is also a concave function.}

Lastly, {given the concavity of $V_{T-1}(e_{T-1})$}, by working recursively, we can conclude that $V_{T-2}(e_{T-2})$, $V_{T-3}(e_{T-3})$,..., $V_{1}(e_{1})$ are all concave, which finishes the proof.

\subsection{Proof of Theorem~\ref{the:ubd}} \label{appendix.gau}
Theorem~\ref{the:ubd} says $v_{t}(0)$ is unbounded if price $\hat \lambda_{t+1}$ follows Gaussian distribution $N(\mu_{t+1}, \sigma_{t+1})$ with fixed mean value $\mu_{t+1}$ and unbounded standard deviation $\sigma_{t+1}$. 

To prove this theorem, we will show that $v_{t}(0)$ is monotonically increasing with respect to $\sigma_{t+1}$ when $\sigma_{t+1} \geq \underline{\sigma}$. Thus, given an arbitrary $\theta$, we can always find a $\sigma_{t+1}$ such that $v_{t}(0) \geq \theta$.

Consider the analytical expression of $q_{t}(e|\hat \lambda_{t+1})$ based on the result in~\cite{zheng2022arbitraging}, {when SoC $e = 0$, }
the analytical expression of $q_{t}(0)$ is 
\begin{align}\label{eq:q0_nik}
    &q_{t}(0|\hat \lambda_{t+1}) = \\
    &\begin{cases}
    v_{t+1}(E_c)  & \text{if $\hat \lambda_{t+1}\leq v_{t+1}(E_c)\eta$} \\
    \hat \lambda_{t+1}/\eta  & \text{if $ v_{t+1}(E_c)\eta < \hat \lambda_{t+1} \leq v_{t+1}(0)\eta$} \\
    v_{t+1}(0) & \text{if $ v_{t+1}(0)\eta < \hat \lambda_{t+1} \leq [\frac{v_{t+1}(0)}{\eta} + c]^+$} \\
    (\hat \lambda_{t+1}-c)\eta & \text{if $ [\frac{v_{t+1}(0)}{\eta} + c]^+ < \hat \lambda_{t+1}$} 
    \end{cases} \nonumber
\end{align}
where curve of $q_{t}(0|\hat \lambda_{t+1})$ with respect to $\hat \lambda_{t+1}$  is sketched in Figure~\ref{fig:q0}.

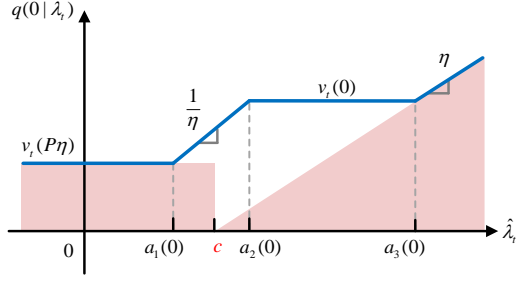
\begin{figure}[h]
    \centering
    \resizebox{0.39\columnwidth}{!}{
        \begin{tikzpicture}
        \draw[->,thick] (-2,0) -- (7,0) node[right] {$\hat \lambda_{t+1}$};
        \draw[->,thick] (0,-0.5) -- (0,4) node[above] {$q_{t}(0|\hat \lambda_{t+1})$};
        
        \draw[domain=-2:-0.2] plot (\x,{1}) node[below left] {$v_{t+1}(E_c)$};
        \draw[domain=-0.2:1] plot (\x,{1});
        
        \draw[domain=1:1.6] plot (\x,{\x}) node[left] {$\frac{1}{\eta}  \hat \lambda_{t+1}$};
        \draw[domain=1.6:2] plot (\x,{\x});
        
        \draw[domain=2:3.3] plot (\x,{2}) node[below] {$v_{t+1}(0)$};
        \draw[domain=3.3:5] plot (\x,{2});
        
        \draw[domain=5:6] plot (\x,{.8*(\x-5)+2}) node[left] {$(\hat \lambda_{t+1}-c)\eta ~$};
        \draw[domain=6:7] plot (\x,{.8*(\x-5)+2});
        
        \filldraw[fill= red!50,opacity=0.3] (-2,0) -- (2.5,0) -- (2.5,1) -- (-2,1) -- cycle;
        \filldraw[fill= red!50,opacity=0.3] (2.5,0) -- (7,0) -- (7,3.6) -- cycle;
        
        \draw[-,thick] (2.5,0.05) -- (2.5,-0.05) node[below] {$c$};
        \end{tikzpicture}
        }
    \caption{Curve of $q_{t}(0|\hat \lambda_{t+1})$. Note that the value of $v_{t+1}(E_c)$ is positive according to Lemma~\ref{lem:ubd_pre}.}
    \label{fig:q0}
\end{figure}

Now we show that $v_{t}(0)$ is  monotonically increasing if $\sigma_{t+1} \geq \underline{\sigma}$, where $\underline{\sigma}$ is a constant at period $t+1$.
Recall that $v_{t}(0)$ is the expectation of $q_{t}(0)$
\begin{align}
    v_{t}(0) & = \int_{-\infty}^{\infty} q_{t}(0|\hat \lambda_{t+1} = x) \cdot f_{\hat \lambda_{t+1}}( x)~d x \nonumber  \\
    &\geq \int_{c}^{+\infty} v_{t+1}(E_c) f_{\hat \lambda_{t+1}}(x) d x \nonumber \\
    &~~ + \int_{c}^{+\infty} \eta(x-c) f_{\hat \lambda_{t+1}}(x) d x \label{the6:eq1} 
\end{align}
where $f_{\hat \lambda_{t+1}}(x) $ is the {Probability Density Function (PDF) function} of price $\hat \lambda_{t+1}$. {$q_{t}(0)$ in~\eqref{eq:q0_nik} is always positive as $v_{t+1}(e)$ is positive in Lemma~\ref{lem:ubd_pre}. } The two integral terms in~\eqref{the6:eq1} correspond to the red areas in Figure~\ref{fig:q0}.

Then by calculating the derivative of the two integral terms in equation~\eqref{the6:eq1} with respect to $\sigma_{t+1}$, we know
\begin{align}
    & \pdv{v_{t}(0)}{\sigma_{t+1}}  \geq  v_{t+1}(E_c) \cdot \frac{\mu_{t+1}-c}{\sigma_{t+1}}f_{\hat \lambda_{t+1}}(c) \nonumber \\ 
    &~~~~~~~~~~~~~-\frac{(\mu_{t+1}-c)(\mu_{t+1}+c)}{\sigma_{t+1}}f_{\hat \lambda_{t+1}}(c)  \nonumber \\ 
    &~~~~~~~~~~~~~+  \frac{(\mu_{t+1}-c)^2}{\sigma_{t+1}}f_{\hat \lambda_{t+1}}(c) + {\sigma_{t+1}}f_{\hat \lambda_{t+1}}(c) \nonumber \\ 
    &~~ = \frac{[(v_{t+1}(E_c)-2c)(\mu_{t+1}-c) + \sigma_{t+1}^2]f_{\hat \lambda_{t+1}}(c)}{\sigma_{t+1}} ~\label{the6:eq2}
\end{align}

If $(v_{t+1}(E_c)-2c)(\mu_{t+1}-c) \geq 0$, it is evident that $v_{t}(0)$ is monotonically increasing with $\sigma_{t+1}$. In this case, $\underline{\sigma} = 0$. 

If $( v_{t+1}(E_c)-2c)(\mu_{t+1}-c) < 0$, condition $\sigma_{t+1} > \sqrt{-(v_{t+1}(E_c)-2c)(\mu_{t+1}-c)}$ is needed to guarantee the right hand side of \eqref{the6:eq2} larger than 0, such that $v_{t}(0)$ monotonically increases. In this case, we obtain that $\underline{\sigma} =\sqrt{-( v_{t+1}(E_c)-2c)(\mu_{t+1}-c)}$. Therefore, we can always find a $\underline{\sigma}$ that guarantees the monotone increasing of $v_{t}(0)$, which finishes the proof.

\subsection{Proof of Lemma~\ref{lem:ubd_pre}}\label{appendix.lemgau}

Given storage duration is long enough such that $2E_c \leq E$ and {the concavity of $V_T(e_T)$, we know that} $v_t(e_t)$ is monotonically decreasing with respect to $e_t$. Then we can conclude that $v_T(0) \geq v_T(E_c) \geq  v_T(2E_c) \geq  v_T(E) \geq 0$. 

Now let $t=T-1$, based on equation~\eqref{eq:q0_nik}, we have
\begin{align}
    v_{t}(E_c) & = \int_{-\infty}^{\infty} q_{t}(E_c|\hat \lambda_t=x) f(x)dx \label{lem3:eq1} \\
    \geq~ &~  v_{t+1}(2E_c) \int_{-\infty}^{v_{t+1}(2E_c)}  f(x)dx \nonumber\\
    & + v_{t+1}(2E_c) \int_{v_{t+1}(2E_c)\eta}^{v_{t+1}(E_c)\eta} f(x)dx \nonumber\\
    & + v_{t+1}(E_c) \int_{v_{t+1}(E_c)\eta}^{v_{t+1}(E_c)/\eta + c} f(x)dx \nonumber\\
    & + v_{t+1}(E_c)\int_{v_{t+1}(E_c)/\eta + c}^{v_{t+1}(0)/\eta + c} f(x)dx \nonumber\\
    & + v_{t+1}(0)\int_{v_{t+1}(0)/\eta + c}^{\infty} f(x)dx \nonumber\\
   \geq ~& v_{t+1}(2E_c) \geq 0, \nonumber
\end{align}
By recursively applying equation~\eqref{lem3:eq1} for time periods $t = T-2, ..., 1$, we establish the inequality $v_t(E_c)\geq 0, \forall t \in \mathcal{T}^{'} = \{1,2,3,...,T-1\}$. Hence we finish the proof.

\subsection{Proof of Corollary~\ref{cor:ubd_general}}\label{appendix.soc}
\begin{subequations}

To prove that the value of $v_t(0)$ can become arbitrarily high with particular price distributions without changing the price expectation, we adopt the finite element method. 
{In particular, while we considered a price distribution model with only two realizations, we can generalize this result to all possible price distributions using the finite element methods by considering any price distribution as the sum of price pairs, then it is trivial to see in any sub-pair distribution models.}  

Now consider a price expectation $\hat \lambda_{t+1}$ and two alternative price values $\pi_{t+1}$ and $\gamma_{t+1}$ such that
\begin{align}
    \pi_{t+1} \geq \mu_{t+1} \geq \gamma_{t+1} \label{lem:eq3}
\end{align}
that satisfy price distribution weights $\alpha_{t+1}>0$ and $\beta_{t+1}>0$ such that
\begin{align}
    \alpha_{t+1} \pi_{t+1} + \beta_{t+1} \gamma_{t+1} &= \mu_{t+1} \label{lem:eq4}\\
    \alpha_{t+1} + \beta_{t+1} &= 1 \;. \label{lem:eq5}
\end{align}
Thus, we consider two price distributions with the same expectation $\mu_{t+1}$: 1) $\hat \lambda_{t+1} = \{\mu_{t+1}|f_{\hat \lambda_{t+1}}(\mu_{t+1}) = 1\}$; 2) $\hat \lambda_{t+1} = \{\pi_{t+1}, \gamma_{t+1}|f_{\hat \lambda_{t+1}}(\pi_{t+1}) = \alpha_{t+1}, f_{\hat \lambda_{t+1}}(\gamma_{t+1}) = \beta_{t+1}\}$. Our goal is thus to prove that given a $\mu_{t+1}$, there exist \textit{a set of} $\pi_{t+1}$, $\gamma_{t+1}$, $\alpha_{t+1}$, and $\beta_{t+1}$ that satisfies \eqref{lem:eq3}--\eqref{lem:eq5} and
\begin{align}\label{lem:eq6}
    v_{t}(0) = \alpha_{t+1} q_{t}(0|\pi_{t+1}) + \beta_{t+1} q_{t}(0|\gamma_{t+1}) \geq \theta
\end{align}
recall $\theta$ is {an} arbitrary bid value.

By observing \eqref{eq:q0_nik} we see that given $\eta\in[0,1]$, $q_{t}(0)$ monotonically increases with $\hat \lambda_{t+1}$. Yet, the lowest value possible is bounded by $v_{t+1}(E_c)$, but the high-value scales with $\hat \lambda_{t+1}$. Thus, we can reasonably set the lower price $\gamma_{t+1}$ to $v_{t+1}(E_c)\eta$, and the high price value $\pi_{t+1}$ is at least greater than $[v_{t+1}(0)/\eta + c]^+$. Then substitute this into \eqref{lem:eq6}, we have
\begin{align}
    \alpha_{t+1} (\pi_{t+1}-c)\eta + \beta_{t+1} v_{t+1}(E_c) &\geq \theta \\
      \frac{\theta-\beta_{t+1} v_{t+1}(E_c)}{\alpha_{t+1}\eta} + c &\leq \pi_{t+1}
\end{align}

Then, summarizing our results, for an arbitrary real value $\theta$, the following price distribution will satisfy $v_{t}(0) \geq \theta$
\begin{align}
    f(\hat \lambda_{t+1}) = \begin{cases}
        \beta_{t+1} \text{ if } \hat \lambda_{t+1} = v_{t+1}(E_c) \\
        \alpha_{t+1} \text{ if } \hat \lambda_{t+1} = \frac{\theta-\beta_{t+1} v_{t+1}(E_c)}{\alpha_{t+1}\eta} + c
    \end{cases}
\end{align}
while ensuring the price expectation is $\mu_{t+1}$ by satisfying \eqref{lem:eq3}--\eqref{lem:eq5}. {Then we use finite element method, which} finishes the proof of this Corollary.
\end{subequations}

\subsection{Proof of Corollary~\ref{cor:ubd_interval}} \label{appendix.ubup}
\begin{subequations}
{Given} $e>E_c$ and a sufficiently high price $\hat \lambda_{t+1} > [v_{t+1}(e-\tau P/\eta)/\eta + c]^+$, we use the same pricing model as in \eqref{lem:eq3}--\eqref{lem:eq6} which ensures the price distribution mean is $\mu_t$:
\begin{align}
        &v_{t}(e) \geq \theta \\
        &\alpha_{t+1} v_{t+1}(e-\tau P/\eta) + \beta_{t+1} v_{t+1}(e+E_c) \geq \theta \\
        &v_{t+1}(e-\tau P/\eta) \geq \frac{\theta-\beta_{t+1} v_{t+1}(e+E_c)}{\alpha_{t+1}} 
        \label{lem2:eq1}
\end{align}
Thus, using the assumed price distribution model to make the value function $v_{t}(e)$  greater than $\theta$, we need to ensure \eqref{lem2:eq1}, which states $v_{t+1}(e-\tau P/\eta)$ must be greater than a given value. Thus, recursively, by moving forward in time, we will reach $e-\tau P/\eta < 0$ in if $t-\kappa \geq e\eta/P$. Then using Proposition~\ref{cor:ubd_general}, we can set the value function at zero SoC to any high values with a price distribution with fixed expectations. 

Now we turn to Proposition~\ref{def:es_bid}, which shows the storage bids proportionally correlate with the marginal value function $v_t$. Hence, we finished the proof.
\end{subequations}

\subsection{Proof of Theorem~\ref{the:bound}}\label{appendix.main2}
We follow the proof framework of Corollary~\ref{cor:ubd_general}. We update our price model to include the price bounds as
\begin{subequations}
\begin{align}
    \overline{\lambda} \geq \pi_{t+1} \geq \mu_{t+1} \geq \gamma_{t+1} \geq \underline{\lambda} \label{the2:eq2}
\end{align}
together the price model \eqref{lem:eq4}--\eqref{lem:eq5}. Using the monotonic decreasing property from Proposition~\ref{pro:concave}, 
we focus on $v_t(0)$, which {indicates the marginal value function at zero SoC. Recall $v_t(e_t)$ is monotonically decreasing.} Then, using \eqref{eq:q0_nik}, starting from the last time period with final SoC opportunity values, we can calculate $v_{t}(0)$ as
\begin{align}
    v_{t}(0) &= \alpha_{t+1} (\pi_{t+1}-c) \eta + \beta_{t+1} v_{t+1}(E_c) \label{the2:eq1}\\
    &= \frac{\mu_{t+1}-\gamma_{t+1}}{\pi_{t+1}- \gamma_{t+1}} (\pi_{t+1}-c) \eta \nonumber
    \\
    &~~~~+ \frac{\pi_{t+1}-\mu_{t+1}}{\pi_{t+1}- \gamma_{t+1}} v_{t+1}(E_c) \nonumber
    \\
    &= (\mu_{t+1}-\gamma_{t+1})\eta(1 + \frac{\gamma_{t+1}-c}{\pi_{t+1}-\gamma_{t+1}}) \nonumber\\
    &~~~~+ v_{t+1}(E_c)(1 +\frac{\gamma_{t+1}-\mu_{t+1}}{\pi_{t+1}-\gamma_{t+1}}) \nonumber\\
    &= (\mu_{t+1}-\gamma_{t+1})\eta + v_{t+1}(E_c) \nonumber\\
    &~~~~ + \frac{(\mu_{t+1}-\gamma_{t+1})\eta(\gamma_{t+1}-c)}{\pi_{t+1} - \gamma_{t+1}} \nonumber \\
    &~~~~ +\frac{(v_{t+1}(E_c)(\gamma_{t+1}-\mu_{t+1})}{\pi_{t+1} - \gamma_{t+1}} \nonumber
\end{align}
Then given $c\geq \gamma_{t+1}$, $\mu_t \geq \gamma_{t+1}$ and $v_{t+1}(E_c) \geq 0$, \eqref{the2:eq1} is an increasing function to $\pi_{t+1}$ and a decreasing function to $\gamma_{t+1}$, hence $v_{t}(0)$ value is maximized if $\pi_{t+1} = \overline{\lambda}$ and $\gamma_{t+1} = \underline{\lambda}$ for all $v_{t+1}(E_c)$. Then we generalize this result to all possible price distributions using the finite element method by considering any price distribution as the sum of price pairs, as in \eqref{the2:eq2}, then it is trivial to see in any sub-pair distribution models. The two price realizations will be at the price bounds. 

According to the previous result, we can consider the following price process at period $t+1$: $f_{\hat \lambda_{t+1}}(\overline{ \lambda}) = \alpha_{t+1}$, $f_{\hat \lambda_{t+1}}(\underline{\lambda}) = \beta_{t+1}$, with $\alpha_{t+1} + \beta_{t+1} = 1$ and the expectation is $\mu_{t+1}$, will maximize $v_{t}(0)$ for any $v_{t+1}(E_c)$. On the other hand, note the temporal transitive property of $q_{t}(0)$ that the value of $v_t(0)$ will be passed to $v_{t+1}(E_c)$ in the previous time periods for high price results. Then $v_{t+1}(E_c)$ becomes the input to the previous zero SoC value function. Hence, \eqref{the2:eq1} forms a recursive calculation with 
\begin{align}
    v_{t}(0) &= \alpha_{t+1} (\pi_{t+1}-c) \eta + \beta_{t+1} v_{t+1}(E_c) \nonumber \\
    &\leq \alpha_{t+1} (\pi_{t+1}-c) \eta + \beta_{t+1} v_{t+1}(0) \nonumber \\
    & = \alpha_{t+1} (\pi_{t+1}-c) \eta  \\
    &~~~~~~~~+ \beta_{t+1} \big[ \alpha_{t+2} (\pi_{t+2}-c) \eta + \beta_{t+2} v_{t+2}(E_c) \big] \nonumber  \\
    &\leq \dotsc \nonumber
\end{align}
which provides the upper bound for the marginal value function $v_{t}(e)$. Then applying equation~\eqref{eq:bid_pc} to design bids gives the result in \eqref{eq:the2}, hence finishing the proof.
\end{subequations}

\subsection{Proof of Proposition~\ref{pro:Lag}} \label{appendix.Lag}

The idea of proving Proposition~\ref{pro:Lag} is to demonstrate the RTM clearing price is monotonically increasing with respect to net demand. 

For clarity, we first denote net demand $D_t - w_t = x_t$, then
\begin{align}
    &L(g_{t}, p_t, b_t, x_t)  \nonumber \\
    &~~=G_t(g_{t}) + \sum_{j=1}^{N_e} O_t(p_t) -  \sum_{j=1}^{N_e} D_t(b_t)  \nonumber \\
    &~~~ + \lambda_t( g_t + p_t - b_t - x_t)   \label{eq:ap_L}
\end{align}
Now we consider net demand $x_t$ is a random variable independent from $g_{t}, p_t, b_t$. Given storage capacity is small enough, we can ignore the terms $ \sum_{j=1}^{N_e} O_t$ and $ \sum_{j=1}^{N_e} D_t$. According to the KKT condition and the derivative chain rule, we know
\begin{align}
    & \pdv{L(g_{t}, p_t, b_t, x_t)}{x_t}  \approx \pdv{G_t(g_{t})}{g_{t}} - \lambda_t = 0
    \label{eq:ap_dL}
\end{align}
 Therefore, we can take the derivative of both sides of equation~\eqref{eq:ap_dL} using the derivative chain rule:
\begin{align}
    \pdv{\lambda_t}{x_t} = \frac{\partial^2 G_t(g_{t})}{\partial g_{t}^2} \geq 0
    \label{eq:ap_dL2}
\end{align}
Considering functions $G_t(g_{t})$ is a convex function for all $t \in \mathcal{T}$, therefore, the second-order derivative is a non-negative constant, indicating the relationship between RTM price $\lambda_t$ and net demand $x_{t}$ is linear, which finishes the proof.

\end{document}